\documentclass[a4paper,twocolumn,10pt,accepted=2020-02-06]{quantumarticle}
\pdfoutput=1

\PassOptionsToPackage{usenames,dvipsnames}{color}
\usepackage{amsmath,amsfonts,amssymb,amsthm,graphicx,bbm,enumerate,times,color}
\usepackage{mathtools,float}
\usepackage{import}
\usepackage{hhline}
\usepackage{hyperref}
\usepackage[numbers]{natbib}
\mathtoolsset{showonlyrefs}

 \newtheorem{theorem}{Theorem}
 \newtheorem{proposition}{Proposition}

 \newtheorem{lemma}[theorem]{Lemma}
 
 \newtheorem{definition}[theorem]{Definition}

\newcommand{\blue}[1]{{\color{blue}#1}}

\newcommand{\red}[1]{{\color{red}#1}}

\newcommand{\mc}[1]{\mathcal{#1}}

\newcommand{\mb}[1]{\mathbb{#1}}

\newcommand{\e}{\mathrm{e}}

\newcommand{\tr}{\mathrm{Tr}} 
\newcommand{\Tr}{\mathrm{Tr}} 

\newcommand{\one}{\mathbf{I}}

\newcommand{\RR}{\mb{R}}

\newcommand{\norm}[1]{\left\Vert #1 \right\Vert}

\newcommand{\ket}[1]{\left.\left|{#1}\right.\right\rangle}

\newcommand{\bra}[1]{\left.\left\langle{#1}\right.\right|}

\newcommand{\ketbra}[2]{\ket{#1} \!\! \bra{#2}}

  \newcommand{\proj}[1]{\ketbra{#1}{#1}}

\newcommand{\fu}{Dahlem Center for Complex Quantum Systems, Freie Universit{\"a}t Berlin, 14195 Berlin, Germany}
\newcommand{\ethz}{Institute for Theoretical Physics, ETH Zurich, 8093 Zurich, Switzerland}
\begin{document}

\title{By-passing fluctuation theorems with a catalyst}

\author{P.~Boes}
\affiliation{\fu}
\orcid{0000-0003-4932-6055}
\author{R.~Gallego}
\affiliation{\fu}
\author{N.~H.~Y.~Ng}
\affiliation{\fu}
\orcid{0000-0003-0007-4707}
\author{J.~Eisert} 
\affiliation{\fu}
\orcid{0000-0003-3033-1292}
\author{H.~Wilming}
\affiliation{\ethz}
\orcid{0000-0002-0306-7679}
\begin{abstract}	
	Fluctuation theorems impose constraints on possible work extraction probabilities in thermodynamical processes. These constraints are stronger than the usual second law, which is concerned only with average values. 
	Here, we show that such constraints, expressed in the form of the Jarzysnki equality, can be by-passed if one allows for the use of catalysts---additional degrees of freedom 
	that may become correlated with the system from which work is extracted, but whose reduced state remains unchanged so that they can be re-used. 
	This violation can be achieved both for small systems but also for macroscopic many-body systems, and leads to positive work extraction per particle with finite probability from macroscopic states in equilibrium. 
	In addition to studying such violations for a single system, we also discuss the scenario in which many parties use the same catalyst to induce local transitions. 
	We show that there exist catalytic processes that lead to highly correlated work distributions, 
	expected to have implications for stochastic and quantum thermodynamics.
\end{abstract}
\maketitle
\section{Introduction}
Consider a physical system in thermal equilibrium with its environment. The second law of thermodynamics dictates that it is impossible to extract positive average work from this system using reversible processes that are cyclic in the Hamiltonian. More precisely, if the system's initial state is represented by a canonical ensemble and we consider many iterations of a probabilistic process during which the Hamiltonian of the system is varied but returned to the initial Hamiltonian at the end, then it holds that 
\begin{equation} 
	\label{eq:second_law}     
	\langle W \rangle \leq 0, 
\end{equation}
where $\langle W \rangle$ is the average work extracted during the process. We will refer to \eqref{eq:second_law} as the \emph{Average Second Law (Av-SL)},

However, there exist significantly stronger constraints on the possible extracted work in the above type of processes, namely those imposed by \emph{fluctuation theorems} \cite{Jarzynski_1997,Crooks1998,Tasaki_2000}. 
Indeed, using such theorems, one can show that the probability of extracting a finite amount of positive work per particle is exponentially suppressed with the number of particles in a system \cite{Jarzynski_1997}. 
Once these different types of constraints are recognized, an interesting questions arises: What are physically meaningful settings in which the probabilistic constraints imposed by fluctuation theorems can be circumvented, while still respecting the Av-SL? In particular, do fluctuation theorems also hold when an additional, cyclically evolving auxiliary system is allowed for?

In this work, we present an answer to this question, by introducing
a class of processes that generalize the above reversible processes, are physically well motivated, compatible with \eqref{eq:second_law}, and yet allow for the extraction of positive work per particle with a probability that is independent of system size. 
We do so via the notion of a \emph{catalytic process}, in which we
allow for the reversible process to not only act on the system as such, but additionally on an auxiliary system that can be initially prepared in an arbitrary state, but whose marginal state has to be left invariant by the process. Such catalysts are well-motivated -- they allow a general description of thermodynamic processes in which the system may be interacting with some experimental apparatus (such as a quantum clock~\cite{erker2017autonomous,woods2019autonomous}), however not extracting energetic/information resources from such an ancilla. In terms of our discussion of the Av-SL above, catalysts correspond to the cyclically evolving auxiliary system. Despite being studied frequently in resource-theoretic formulations of thermodynamics~\cite{Brandao2015,Ng2015,Mueller2017,Boes2018a}, catalytic processes have never been studied in the context of fluctuation theorems until now. Furthermore, even in previous works of catalysis, the exact form of the catalyst is highly state-dependent and therefore rarely studied explicitly~\cite{Brandao2015,Mueller2017}.
In this work, we make progress in the significant gaps in the knowledge of catalysis, by presenting and discussing constructive examples of such catalytic processes in the framework where fluctuation theorems are commonly derived. We show that, by sharing the same catalyst, a group of agents can follow collective strategies to achieve highly correlated work-distributions. 
This makes these processes interesting for the field of \emph{quantum and stochastic thermodynamics} and potentially also for certain negentropic processes in biology. On the overall, our work provides a rigorous footing for the further study of thermodynamical processes that systematically exploit the notion of \emph{catalysis}
in order to achieve certain patterns of work fluctuations in an environment that is governed by the Av-SL. 
Given the broad applicability of our results, we believe that the study of such processes will produce many further interesting results of both foundational and practical interest.

\section{Setup}

\subsection{Formulation of the physical situation} 
We formulate our arguments and results in the language of quantum mechanics, but all of our results similarly apply to classical, stochastic systems.
We consider the setting depicted in Fig.~\ref{fig:setup}: A $d$-dimensional system $S$ with Hamiltonian $H = \sum_{i=1}^d E_i \proj{E_i}$ is initalized in the Gibbs state
\begin{equation}
	\omega_\beta(H) := \frac{\e^{-\beta H}}{Z(\beta, H)}, 
\end{equation}
where $Z(\beta, H) := \tr(\e^{-\beta H})$. 
This state describes a system initially in thermal equilibrium with its environment 
at inverse temperature $\beta := 1/(k_B T)$.
An agent (some experimenter) first performs an energy measurement on this system 
which produces a measurement outcome $ E_i $. According to quantum mechanics, the post-measurement state is described by the density matrix $\ketbra{E_i}{E_i}$. The agent then performs a physical operation on the system which does not depend on the outcome of the measurement. Such an operation can always be represented by a general quantum channel $\mc C$ (i.e., a trace-preserving, completely positive map that takes density matrices to density matrices) applied to the post-measurement state. This operation is then followed by a second energy measurement with respect to the same Hamiltonian with outcome $E_f$ \footnote{It is possible to extend the setup and our further results to the more general case of different Hamiltonians for the initial and final measurement. We present our results within this restricted settings for conceptual and notational simplicity.}.  
This procedure results in a channel-dependent joint distribution ${P(E_f, E_i) = P(E_f|E_i) P(E_i)}$. 
In general, a given quantum channel may be realized in different ways. Whether the change of energy $E_f-E_i$ can be interpreted as work from a thermodynamic point of view will depend on how exactly the quantum channel $\mc C$ was physically realized. We will assume that this is the case in the following, but will comment on this assumption again later on.
In particular, we can then define the work distribution $P$ for the above process as 
\begin{equation} \label{eq:work_distribution}
	P(W) := \sum_{i,f} P(E_f, E_i) \delta(W - (E_i - E_j)), 
\end{equation}
where $\delta$ is the Dirac delta distribution.
We are interested in investigating possible distributions $ P(W) $
that arise from different channels $\mc C$. To do so, it is useful to note the relation
\begin{equation} \label{eq:general_jar}
	\langle \e^{\beta W} \rangle = \sum_{j} \frac{e^{-\beta E_j}}{Z_H} \bra{E_j} \mc{C} [\one] \ket{E_j}, 
\end{equation}
which is straightforwardly derived using the above definitions, where $\one$ denotes the identity matrix. 

\begin{figure}[t]
	\centering
	\includegraphics[width=0.48\textwidth]{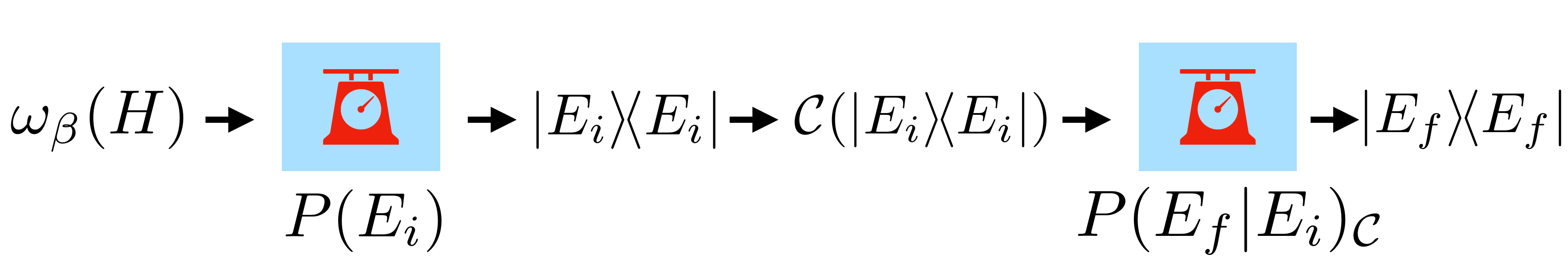}
	\caption{The basic setup for all processes in this work: An agent with access to a system $S$ 
		equipped with Hamiltonian $H$ that is assumed to be initially in thermal equilibrium with a heat bath at inverse temperature $\beta$ samples from $S$ (by measuring in the energy basis), then implements a process that maps the post-measurement state $\proj{E_i}$ to $\mc C(\proj{E_i})$, where $\mc C$ is a quantum channel. Finally, the agent repeats the energy measurement on $S$ with respect to the same Hamiltonian $H$. 
	}
	\label{fig:setup}
\end{figure}

In the standard setting of \emph{Tasaki-type fluctuation theorems}, $\mc C$ is considered to be a unitary channel $\mc C[\cdot] = U (\cdot) U^\dagger$, since these are generated by changing the Hamiltonian over time~\cite{Tasaki_2000}.  
For such channels, \eqref{eq:general_jar} becomes 
\begin{equation}\label{eq:jar}
	\langle \e^{\beta W} \rangle = 1, 
\end{equation} 
which is the well-known Jarzynski equality (JE) for cyclic, reversible processes \cite{Jarzynski_1997}. Eq. \eqref{eq:jar} is strictly stronger than \eqref{eq:second_law}, the latter being implied by \eqref{eq:jar} via Jensen's inequality.

\subsection{No macroscopic work}
One of the reasons for the importance of the JE derives from the fact that it gives strong bounds on the possibility of extracting work from a large system in a thermal state~\cite{jarzynski2011equalities,cavina2016optimal,maillet2019optimal}.
To see this, let $S$ be an $N$-particle system and define the probability of extracting work $w$ per particle as 
\begin{equation}
	p(w) := P(wN). 
\end{equation}
Plugging this into \eqref{eq:jar} yields that for any $ \epsilon >0 $,
\begin{equation}
	1 = \langle \e^{\beta W} \rangle =\sum_{w} \e^{\beta wN} \: P(wN) \geq \e^{\beta \epsilon N} \sum_{w\geq\epsilon}\: p(w), 
\end{equation}
which implies that events which extract significant positive work per particle from a macroscopic system at equilibrium are exponentially unlikely in $ N $. 
For later use, we formalize this property.

\begin{definition}[No macroscopic work]\label{dwf:nmw} Given a sequence of $N$-particle systems initially at thermal equilibrium with inverse temperature $\beta$ and channels $\mc C$ (implicitly depending on $N$), we say that the processes represented by $\mc C$ fulfill the \emph{no macroscopic work (NMW)} condition if the probability of an event extracting work per particle larger or equal than $\epsilon$ is arbitrarily small as $N\rightarrow \infty$,
	\begin{equation}\label{eq:probabilistic_secondlaw}
		\lim_{N\rightarrow \infty }p(w\geq \epsilon) := \lim_{N\rightarrow \infty}\sum_{w>\epsilon}\: p(w) = 0. 
	\end{equation}
\end{definition}
As is clear from the above, channels that satisfy the JE, such as unitary channels, also satisfy NMW and Av-SL. We now turn to investigate violations of JE and NMW for non-unitary channels.

\section{Violations of NMW and JE}
The first main result of this work is to introduce a physically motivated family of channels $\mc C$ that violates both NMW and JE, but respects the Av-SL. To aid comparison, we first briefly discuss other generalizations of the standard setting to non-unitary channels (see also Refs.~\cite{Rastegin_2013,Rastegin_2014}). 
\subsection{ Violating JE with non-unitary channels}
It is easy to see from \eqref{eq:general_jar} that a more general class of channels that satisfy the JE are \emph{unital} channels, that is, channels that satisfy $\mc C[\one] = \one$. Consequently, neither JE, nor in turn NMW or Av-SL can be violated in settings which give rise to a unital channel.
However, once this condition on unitality is relaxed, it becomes easy to violate JE on a formal level. For example, consider the fully-thermalizing channel that maps every input state to the thermal state $\omega_\beta(H)$, in other words $ \mathcal{C} (\cdot) = \omega_{\beta}(H) $. This channel always violates the JE whenever $ \omega_{\beta}(H)\neq \mathbf{I}/d $. It is, however, not clear how the energy-fluctuations can be interpreted as \emph{work} in this example, since thermalizing processes usually occur due to contact with a heat bath, in which case one would naturally interpret the changes of energy on the system being due to heat. Thus, while it is trivial to formally violate JE, it is not obvious whether it is possible to do so in a physically meaningful and operationally useful manner.
Nevertheless, in Appendix \ref{app:nmw_for_gp}, we show that the fully-thermalizing channel, in fact any channel with the thermal state as a fixed point, cannot violate the NMW condition for typical many-body systems, even if they may violate \eqref{eq:jar}. This means that, even if one interprets energy fluctuations as work, one still could not use the thermalizing channel to extract macroscopic amounts of work from a many-body system. 

\subsection{Violations of NMW and JE via $\beta$-catalytic channels}
The above findings raise the important question whether there exist channels for which the above procedure leads to a violation of NMW (and hence JE), while still respecting the Av-SL and allowing for the interpretation of the random variable $W$ as work extracted from $S$.
Such channels, if they exist, promise to be of great interest because they could allow for a systematic exploitation of relatively likely events extracting work from heat baths.
The first result of this work is to answer this question affirmatively.
To this end, we define the notion of a $\beta$-catalytic channel.

\begin{definition}[$\beta$-catalytic channel] \label{def:catalytic}
	A completely positive, trace-preserving map $\mc C$ is a \emph{$\beta$-catalytic channel} on S, if there exists a quantum state $\sigma_C$ on a system $C$ with  Hamiltonian $H_C$, together with a unitary $U$ such that $[\sigma_C,H_C]=0$ and 
	\begin{align}\label{eq:channel_dilation}
		\mc{C}(\cdot) = \tr_{C}(U (\: \cdot \otimes \: \sigma_C )U^\dagger),                                                  \\
		\label{eq:catalytic_condition}	\text{s.t.} \: \: \tr_{S}( U (\omega_{\beta}(H) \otimes \sigma_C )U^\dagger)=\sigma_C. 
	\end{align}
\end{definition}
Before stating our first main result, let us make some comments about this definition. First of all, we already assumed that the initial and final Hamiltonian coincides. This means that while during the process, $ \mathcal{C} $ may couple system and catalyst for example by introducing interaction terms $ H_{SC} $, nevertheless at the end of the process, the channel must also turn off such interaction terms. Secondly, note that $\beta$-catalytic channels describe reversible processes, in the sense that they do not change the entropy of the joint-system $SC$ and can be undone by acting on this joint-system by a unitary process. We refer to the system $C$ as being the ``catalyst'', understanding that it may be some by-stander system involving additional degrees of freedom. 
This terminology is motivated by the fact that, on average, i.e., if we do not condition on the outcomes of the energy measurements, then $C$ is returned, at the end of the procedure, to its original state. It can therefore be re-used for further rounds of the protocol with \emph{new} copies of $S$. Note, however, that the invariance of the reduced state on $C$ under the channel is required not for all initial states of $ S $, but \emph{only} for $\omega_{\beta}(H)$. As such, $\beta$-catalytic channels depend on $\beta$ and $H$ through the second condition. 

While Definition~\ref{def:catalytic} does not require the catalyst to be uncorrelated with $S$ at the end of the protocol, and in this sense goes beyond the conventional notion of catalysis discussed in the resource-theoretic literature on quantum thermodynamics \cite{Brandao2015,Ng2015}, the more general notion of catalysis that we employ here is receiving increasing interest in quantum thermodynamics, where it was shown to single out the quantum relative entropy, free energy and von Neumann entropy \cite{Wilming2017a,Mueller2017,Boes_2018}, to be useful in the context of algorithmic cooling \cite{Boes_2018,Alhambra_2019} and to show the energetic instability of passive states \cite{Sparaciari2017}. 
Finally, let us briefly comment on the interpretation of the random variable $W$ as work in the setting of $\beta$-catalytic channels and the role of the Hamiltonian of the catalyst. Since the process on $C$ and $S$ is unitary, 
it is meaningful to denote the total changes of energies of the two systems as work measured by a two-point measurement scheme on each system. 
This gives rise to a joint-distribution of work on the two systems $P(W^{(S)},W^{(C)})$. The probability distribution of work $P(W)$ discussed above then simply corresponds to the marginal distribution $P(W^{(S)})$ on $S$. Importantly, this distribution is independent of the Hamiltonian on $C$ (see Sec.~\ref{sec:nontrivialH} in the Appendix).
In particular, we can assume that the catalyst has trivial Hamiltonian $H_C=0$, which 
in turn implies $[\sigma_C,H_C]=0$ for any $\sigma_C$. 
It is then clear that no energy flows from the catalyst to the system, not even probabilistically.
For the rest of the article, we hence assume that $H_C=0$.

Given these constraints, it may, at first glance, be unclear how such a catalyst would offer any advantage to violating JE. For instance, one apparent way to make use of the catalyst is to perform a controlled unitary on $ S $, conditioned on $ C $: For some $ \sigma_C = \sum_i p_i \ketbra{i}{i} $, one uses a unitary in Eq.~\eqref{eq:catalytic_condition} of the form 
\begin{equation}
	U_{\rm SC}:=\sum_i U_i \otimes \ketbra{i}{i}_{\rm C}.
\end{equation}
This special case of $\beta$-catalytic channels by construction produces random unitary channels~\cite{audenaert2008random,Boes_2018} on $ S $, which have the form $ \mathcal{C}_{\rm RU} (\cdot) = \sum_i p_i U_i (\cdot) U_i^\dagger $. But random unitary channels are always unital, and therefore automatically satisfy JE. 

In the following, we show that there exist non-unital $\beta$-catalytic channels that allow for a meaningful 
violation of both NMW and JE, while at the same time they always respect the Av-SL. To see the latter, 
we note that
these channels necessarily increase the von Neumann entropy of the input Gibbs state. This follows 
from the sub-additivity of entropy and the fact that $C$ remains locally unchanged. 
Now, since $\omega_{\beta}(H)$ is the state with the least energy given a fixed entropy
\cite{Lenard1978,pusz1978}, then we also have that 
\begin{equation}
	\tr( H \mc{C}(\omega_{\beta}(H)) \geq  \tr( H \omega_{\beta}(H)) 
\end{equation}
which is just the Av-SL, concomitant with the findings of Ref.\ \cite{Boes2018a}. 
We stress that despite this property, $\beta$-catalytic channels are in general \emph{not}
unital. 
It remains to be shown that $\beta$-catalytic channels that violate JE and NMW do exist. We first show that JE 
can be violated already with small quantum systems, and then turn to the violation of NMW for 
macroscopic many-body systems with physically realistic Hamiltonians.\\

\emph{Microscopic violation of JE.} As a toy-like example of violating the JE with $\beta$-catalytic channels, we consider a system with three states -- two degenerate (but distinguishable) ground states and an excited state with energy $E$. As catalyst, we consider a system with two states and the unitary is a simple permutation between two pairs of energy eigenvalues of the joint system (for details, see App.~\ref{app:microscopic_toy_example}).
It is straightforward to compute the probability distribution of work for such small systems, which in this case leads to
\begin{equation}
	\langle \e^{\beta W}\rangle = \frac{Z+5+2(Z-2)(Z-1)}{Z(Z+1)} \geq 1, 
\end{equation}
where $Z=2+\e^{-\beta E}$ is the partition function of the system and we used $2\leq Z \leq 3$. We hence find $\langle \e^{\beta W}\rangle >1$ whenever $E>0$ (since then $Z<3$) and we obtain a moderate maximum violation in the limit $E\rightarrow \infty$ given by $\langle \e^{\beta W} \rangle = 7/6$.\\ 

\emph{Macroscopic violation of NMW condition.} 
We now show that one can violate the NMW principle using catalysts.

\begin{proposition}[Violation of no macroscopic work with catalysts] \label{prop:single_player}
	Let $(S^{(N)})_N$ be a sequence of $N$-particle locally interacting lattice systems with Hamiltonian $H^{(N)}$ that satisfy mild assumptions. Then, for sufficiently large $N$, there exist values of $\epsilon>0$, such that
	\begin{equation}\label{eq:violationjarzynski}
		p(w \geq \epsilon) 
	\end{equation}
	can be brought arbitrarily close to $\frac{1}{2}$ with $\beta$-catalytic channels.
\end{proposition}
We provide a proof and full statement of the assumptions in Appendix~\ref{app:canonical}. Our assumptions are satisfied by typical many-body Hamiltonians with energy windows in which the density of states grows exponentially \cite{huang1987statistical}.

While the formal proof of Proposition~\ref{prop:single_player} is given in the Appendix, the idea behind it is simple and we sketch it here on a higher level. For a given $N$, let $e^{(N)}$ denote the mean energy per particle of an $N$-particle system that satisfies our assumptions. In the proof, we show that for systems that satisfy the above assumptions and any $\delta > 0$, there exists an $N$ and a $\beta$-catalytic channel $\mc C$ such that 
\begin{equation} \label{eq:proof_sketch}
	\mc C(\omega_\beta) \approx_\delta \frac{1}{2} \proj{E_-} + \frac{1}{2} \tau,
\end{equation}
where $\approx_\delta$ denotes equality of the states on LHS and RHS up to $\delta$ in trace distance, $\ket{E_-}$ 
is some eigenvector of $H$ with $E_- < e^{(N)}N$ and $\tau$ is some other ``fail''-state the details of which are irrelevant. We can interpret Eq.~\eqref{eq:proof_sketch} as describing the approximation of a work extraction protocol that results in the state $\ket{E_-}$ with probability $1/2$. Now, as the result of standard concentration bounds, for large $N$ the mass of the thermal state $\omega_\beta$ will be highly concentrated around energy $e^{(N)}N$. This implies that every time the above work extraction protocol succeeds to prepare the ground state, for sufficiently high values of $N$ the extracted work per particle is arbitrarily close to $\epsilon \equiv e^{(N)} - E_-/N$, leading to the statement of Prop.~\ref{prop:single_player}. 

We note that it is remarkable that 
catalytic channels, which are guaranteed to satisfy the Av-2nd law, 
allow for the preparation of states like the one described in Eq.~\eqref{eq:proof_sketch}, 
in which a pure low-energy state carries much of the weight, from a thermal state. 
Indeed, it has recently been conjectured that with the help of catalysts \emph{any} state transition between full-rank states that increases the entropy is possible \cite{Boes2018a}, a statement known as the \emph{catalytic entropy conjecture}. Prop.~\ref{prop:single_player}, and in particular the ability to prepare the state in Eq.~\eqref{eq:proof_sketch}, further supports this conjecture, which has not been proven so far (even though strong evidence has been established).

Similar results as above also apply to the case in which the initial state of the system is described by a \emph{micro-canonical ensemble} rather than the Gibbs state, highlighting a similar contrast to fluctuation theorem results in the micro-canonical regime \cite{PhysRevE.77.051131}. For detailed discussions and proves of corresponding statements in this regime, see Appendix~\ref{app:microcanonical}.

One may wonder whether the creation of correlations between system and catalyst is in fact necessary to violate the NMW principle. This is indeed true, when one simply forces the catalyst to remain uncorrelated in the definition of $\beta$-catalytic channels. A proof of this statement along with further discussion on this problem can be found in Appendix~\ref{app:trumping}. Interestingly, such processes at the same time allow for a violation of the Jarzynski equality. A particular example is given by the fully thermalizing channel, which can be realized using a catalyst that is simply a copy of the Gibbs state of the system and the unitary simply swapping the system and catalyst.

\emph{Required size of the catalyst. }
Proposition~\ref{prop:single_player} not only shows that there exist catalytic procedures that allow an agent to bypass the work extraction bounds imposed by the JE -- the violation of JE is in fact exponential in the system size. In particular, \eqref{eq:violationjarzynski} implies that there exist values $\epsilon > 0$, such that
\begin{equation} \label{eq:unbounded_violation}
	\langle \mathrm e^{\beta W}\rangle  \: \geq \frac{1}{2} \mathrm e^{\beta N\epsilon} \gg 1 
\end{equation}
in the limit of large $N$. 
It is natural to wonder how far the JE can be violated and how big the catalyst has to be to realize a certain violation. 
This is clarified by the following result.
\begin{proposition}[Bound on violation of JE] \label{prop:cat_bound}
	Let $\mc C$ be any $\beta$-catalytic channel with $d_C = \dim(H_C)$. Then,
		 
	\begin{align}
		\langle \mathrm{e}^{\beta W} \rangle & \leq \mathrm{min}\{ d_C \norm{\sigma}_\infty, d \norm{\omega_\beta(H)}_\infty \} \\
		                                     & \leq \mathrm{min}\{d_C,d\},                                                      
	\end{align}
	where $\| \cdot \|_\infty$ denotes the $\infty$-norm, which, for density matrices, equals the largest absolute value of the input's eigenvalues.
\end{proposition}
This proposition, the simple proof of which is given in Appendix \ref{app:lower_bound_on_catalyst_s_dimension}, shows that in order to extract a growing amount of work from a single run of a process, an external agent will have to be able to prepare a state $\sigma$ on a growing auxiliary system and, more importantly, also have control over the increasingly large joint system. 
Hence, in practice, the ability to violate JE will still be constrained by operational limitations.
To illustrate the implications of Prop.~\ref{prop:cat_bound}, let us show how it immediately implies a bound on $P(W)$. As noticed when deriving the NMW principle, for any $\epsilon\geq 0$ we have
\begin{equation}
	\langle \e^{\beta W}\rangle \geq P(W\geq \epsilon) \e^{\beta \epsilon}. 
\end{equation}
Hence, Prop.~\ref{prop:cat_bound} implies
\begin{equation}
	P(W\geq \epsilon) \leq d_C\norm{\sigma}_\infty \e^{-\beta\epsilon}. 
\end{equation}
In particular this means that to extract a macroscopic amount of work, $W\geq w N$, with finite probability, $d_C$ has to grow exponentially with $N$ (note that $\norm{\sigma}_\infty\leq 1$).

\section{Multi-partite work extraction}
As emphasized before, even though the state of the catalyst remains unchanged in a catalytic process, in general it builds up correlations with the system.
We now show that the correlations established between catalyst and system allow for processes in which many agents re-use the same catalyst to obtain highly inter-correlated work distributions.

Consider $n$ agents, each with identical systems $S_i, i \in \{1, \dots, n\}$ that are initialized in the Gibbs state $\omega(\beta,H)$. For a given $\beta$-catalytic channel $\mc C$ with state $\sigma$ on the catalyst, consider the following protocol: Agent $1$ runs the standard process from Fig.~\ref{fig:setup} using the catalyst and hence implementing $\mc C$ between the two measurements. After the procedure, she then passes $C$ on to agent $2$ who repeats this process, and so on, until the last agent has received $C$ and performed the process. From the catalytic nature of $\mc C$, is is clear that, for each agent, the same marginal distribution of work is obtained. However, the joint work distribution for all agents will be correlated, due to individual correlations between each $S_i$ with $C$. 
We now show that the agents can use these correlations to systematically achieve certain global work distributions. Using the same notation as before, let $p(w_1, \dots, w_n)$ denote the global distribution over the extracted work per particle, assuming that all $S_i$ are copies of the same $N$-particle system. We have the following, proven in Appendix~\ref{app:many_player_strategies}.

\begin{proposition}[Multiple agents] \label{prop:many_agents}
	Let each $\lbrace S_i\rbrace_{i=1}^n$ be a sequence of N-particle systems that satisfy the conditions of Proposition~\ref{prop:single_player}. Then, for sufficiently large $N$, there exists an $\epsilon > 0$, such that
	\begin{equation}\label{eq:violation_jarzynski}
		\begin{split}
			p(\epsilon,-\epsilon, \epsilon, -\epsilon, \dots) &= \lambda, \\
			p(- \epsilon,\epsilon, -\epsilon, \epsilon, \dots) &= 1- \lambda,
		\end{split}
	\end{equation}
	where $\lambda$ can be brought arbitrarily close to $1/2$ using a sequence of $\beta$-catalytic channels on $S_i$ and $C$.
\end{proposition}

While \eqref{eq:violation_jarzynski} is clearly consistent with \eqref{eq:second_law}, this proposition shows that the agents can achieve joint work distributions that are strongly correlated and in which subsets of agents, in the above proposition one half of them, can violate JE arbitrarily, at the cost of the other half. 
Such distributions of work could, for example, be of interest in situations where 
the target is to maximize the probability that a subset of players extracts a positive amount work, at the ready cost of the others, for instance in order to surpass an activation energy. 
Importantly, the size of the catalyst needed to realize the distribution \eqref{eq:violation_jarzynski} is fixed, i.e., it does not scale with the number of agents $n$.

Proposition~\ref{prop:many_agents} shows the existence of catalytic processes that produce very interesting global work distributions. This naturally raises the question what other global distributions can be obtained in a setting without making the size of the catalyst depend on the number of rounds. Our results, however, already imply that not every distribution compatible with the Second Law can be obtained in such a way. For instance, Proposition~\ref{prop:cat_bound} implies that the distribution 
\begin{equation} \label{eq:all_same_dist}
	p(\epsilon,\epsilon, \epsilon, \epsilon, \dots) = p(- \epsilon,-\epsilon, -\epsilon, -\epsilon, \dots) \approx 1/2 
\end{equation}
cannot be obtained via $\beta$-catalytic channels, 
since otherwise there would exist a catalyst of fixed size that would allow, for any $n$, the total work $W=n\epsilon$ to be extracted with probability approximately $1/2$, in violation of Proposition~\ref{prop:cat_bound}.

\section{Summary and future work.} In this work we have studied work extraction protocols from states at thermal equilibrium. We significantly expand the common setting of fluctuation theorems under cyclic, 
reversible processes by introducing a catalyst---an additional system which, on average, remains unchanged after the protocol and can thus be re-used. 
This extension enables for distributions of work extraction that are not attainable without a catalyst. More precisely, one can bypass the stringent conditions imposed by the JE, achieving positive work per particle with high probability, even for macroscopic systems.
Furthermore, it allows for interesting, correlated work distributions when many agents use the same catalyst. 

Our constructions illustrate in a striking way that the absence of correlations, sometimes referred to as `stochastic independence', can also be a powerful thermodynamic resource \cite{Pastena}. This complements findings where the initial presence of correlations between a system and an ancilla are used to bypass the standard constraints imposed by fluctuation theorems \cite{sagawa2012fluctuation,sagawa2013role}. We discuss the connection of our work to these findings in more detail in Appendix \ref{app:sagawa}.
We believe that the further study of work distributions that can be obtained by collaborating agents by means of $\beta$-catalytic channels will yield both foundational and practical insights.

We further believe that it is an interesting open problem to study how the size of the catalyst has to scale if one wishes to maximize the probability to extract a certain amount of work.
For example, in the context of a many-body system one might be content with extracting only an amount of work of the order of $\sqrt{N}$ if in exchange for that one can either increase the probability for it to happen significantly or can reduce the size of the catalyst considerably (and hence the complexity of the unitary required to be implemented).

It would be interesting to understand the relation between our results and a more generalized type of JE in the presence of information exchange \cite{PhysRevLett.109.180602}, for example in a Maxwell demon scenario. In particular, in Ref.\ \cite{toyabe2010experimental} it was also demonstrated that by using feedback control, one may also violate JE while respecting the Av-SL. 
More generally, our results also raise the question whether other phenomena --usually described as forbidden by the second law, or as occurring with vanishing probability-- can be made to occur with high probability using catalysts. For example, is it possible to reverse the mixing process of two gases or induce heat flow from a cold to a hot system with finite probability in macroscopic systems? The techniques developed in this work provide a promising ansatz for the study of this and similar questions.

\emph{Acknowledgements.} 
\label{sec:acknowledgements}
We thank Markus P. M{\"u}ller and Alvaro M. Alhambra for valuable discussions and anonymous referees for interesting comments. P.~B. acknowledges support from the John Templeton Foundation. H.~W.\ acknowledges support from the Swiss National Science Foundation through SNSF project No. 200020\_165843 and through the National Centre of Competence in Research \emph{Quantum Science and Technology} (QSIT). N.~H.~Y.~N.\ acknowledges support from the Alexander von Humboldt Foundation. 
R.~G.\ has been supported by the DFG (GA 2184/2-1). J.~E.\ acknowledges support by the DFG (FOR 2724), 
dedicated to quantum thermodynamics, and the FQXi.

\bibliographystyle{apsrev4-1}
\bibliography{catalytic_fluct.bib_doi}

\begin{thebibliography}{39}%
\makeatletter
\providecommand \@ifxundefined [1]{%
 \@ifx{#1\undefined}
}%
\providecommand \@ifnum [1]{%
 \ifnum #1\expandafter \@firstoftwo
 \else \expandafter \@secondoftwo
 \fi
}%
\providecommand \@ifx [1]{%
 \ifx #1\expandafter \@firstoftwo
 \else \expandafter \@secondoftwo
 \fi
}%
\providecommand \natexlab [1]{#1}%
\providecommand \enquote  [1]{``#1''}%
\providecommand \bibnamefont  [1]{#1}%
\providecommand \bibfnamefont [1]{#1}%
\providecommand \citenamefont [1]{#1}%
\providecommand \href@noop [0]{\@secondoftwo}%
\providecommand \href [0]{\begingroup \@sanitize@url \@href}%
\providecommand \@href[1]{\@@startlink{#1}\@@href}%
\providecommand \@@href[1]{\endgroup#1\@@endlink}%
\providecommand \@sanitize@url [0]{\catcode `\\12\catcode `\$12\catcode
  `\&12\catcode `\#12\catcode `\^12\catcode `\_12\catcode `\%12\relax}%
\providecommand \@@startlink[1]{}%
\providecommand \@@endlink[0]{}%
\providecommand \url  [0]{\begingroup\@sanitize@url \@url }%
\providecommand \@url [1]{\endgroup\@href {#1}{\urlprefix }}%
\providecommand \urlprefix  [0]{URL }%
\providecommand \Eprint [0]{\href }%
\providecommand \doibase [0]{http://dx.doi.org/}%
\providecommand \selectlanguage [0]{\@gobble}%
\providecommand \bibinfo  [0]{\@secondoftwo}%
\providecommand \bibfield  [0]{\@secondoftwo}%
\providecommand \translation [1]{[#1]}%
\providecommand \BibitemOpen [0]{}%
\providecommand \bibitemStop [0]{}%
\providecommand \bibitemNoStop [0]{.\EOS\space}%
\providecommand \EOS [0]{\spacefactor3000\relax}%
\providecommand \BibitemShut  [1]{\csname bibitem#1\endcsname}%
\let\auto@bib@innerbib\@empty
\bibitem [{\citenamefont {Jarzynski}(1997)}]{Jarzynski_1997}%
  \BibitemOpen
  \bibfield  {author} {\bibinfo {author} {\bibfnamefont {C.}~\bibnamefont
  {Jarzynski}},\ }\href {\doibase 10.1103/PhysRevLett.78.2690} {\bibfield
  {journal} {\bibinfo  {journal} {Phys. Rev. Lett.}\ }\textbf {\bibinfo
  {volume} {78}},\ \bibinfo {pages} {2690} (\bibinfo {year}
  {1997})}\BibitemShut {NoStop}%
\bibitem [{\citenamefont {Crooks}(1998)}]{Crooks1998}%
  \BibitemOpen
  \bibfield  {author} {\bibinfo {author} {\bibfnamefont {G.~E.}\ \bibnamefont
  {Crooks}},\ }\href {\doibase 10.1023/A:1023208217925} {\bibfield  {journal}
  {\bibinfo  {journal} {J. Stat. Phys.}\ }\textbf {\bibinfo {volume} {90}},\
  \bibinfo {pages} {1481} (\bibinfo {year} {1998})}\BibitemShut {NoStop}%
\bibitem [{\citenamefont {Tasaki}(2000)}]{Tasaki_2000}%
  \BibitemOpen
  \bibfield  {author} {\bibinfo {author} {\bibfnamefont {H.}~\bibnamefont
  {Tasaki}},\ }\href@noop {} {\bibfield  {journal} {\bibinfo  {journal} {ArXiv
  e-prints}\ } (\bibinfo {year} {2000})},\ \Eprint
  {http://arxiv.org/abs/1303.6393} {arXiv:1303.6393} \BibitemShut {NoStop}%
\bibitem [{\citenamefont {Erker}\ \emph {et~al.}(2017)\citenamefont {Erker},
  \citenamefont {Mitchison}, \citenamefont {Silva}, \citenamefont {Woods},
  \citenamefont {Brunner},\ and\ \citenamefont {Huber}}]{erker2017autonomous}%
  \BibitemOpen
  \bibfield  {author} {\bibinfo {author} {\bibfnamefont {P.}~\bibnamefont
  {Erker}}, \bibinfo {author} {\bibfnamefont {M.~T.}\ \bibnamefont
  {Mitchison}}, \bibinfo {author} {\bibfnamefont {R.}~\bibnamefont {Silva}},
  \bibinfo {author} {\bibfnamefont {M.~P.}\ \bibnamefont {Woods}}, \bibinfo
  {author} {\bibfnamefont {N.}~\bibnamefont {Brunner}}, \ and\ \bibinfo
  {author} {\bibfnamefont {M.}~\bibnamefont {Huber}},\ }\href {\doibase
  10.1103/PhysRevX.7.031022} {\bibfield  {journal} {\bibinfo  {journal} {Phys.
  Rev. X}\ }\textbf {\bibinfo {volume} {7}},\ \bibinfo {pages} {031022}
  (\bibinfo {year} {2017})}\BibitemShut {NoStop}%
\bibitem [{\citenamefont {Woods}\ \emph {et~al.}(2019)\citenamefont {Woods},
  \citenamefont {Silva},\ and\ \citenamefont
  {Oppenheim}}]{woods2019autonomous}%
  \BibitemOpen
  \bibfield  {author} {\bibinfo {author} {\bibfnamefont {M.~P.}\ \bibnamefont
  {Woods}}, \bibinfo {author} {\bibfnamefont {R.}~\bibnamefont {Silva}}, \ and\
  \bibinfo {author} {\bibfnamefont {J.}~\bibnamefont {Oppenheim}},\ }\href
  {\doibase 10.1007/s00023-018-0736-9} {\bibfield  {journal} {\bibinfo
  {journal} {Ann Hen. Poin.}\ }\textbf {\bibinfo {volume} {20}},\ \bibinfo
  {pages} {125} (\bibinfo {year} {2019})}\BibitemShut {NoStop}%
\bibitem [{\citenamefont {Brand{\~{a}}o}\ \emph {et~al.}(2015)\citenamefont
  {Brand{\~{a}}o}, \citenamefont {Horodecki}, \citenamefont {Ng}, \citenamefont
  {Oppenheim},\ and\ \citenamefont {Wehner}}]{Brandao2015}%
  \BibitemOpen
  \bibfield  {author} {\bibinfo {author} {\bibfnamefont {F.~G. S.~L.}\
  \bibnamefont {Brand{\~{a}}o}}, \bibinfo {author} {\bibfnamefont
  {M.}~\bibnamefont {Horodecki}}, \bibinfo {author} {\bibfnamefont {N.~H.~Y.}\
  \bibnamefont {Ng}}, \bibinfo {author} {\bibfnamefont {J.}~\bibnamefont
  {Oppenheim}}, \ and\ \bibinfo {author} {\bibfnamefont {S.}~\bibnamefont
  {Wehner}},\ }\href {\doibase 10.1073/pnas.1411728112} {\bibfield  {journal}
  {\bibinfo  {journal} {PNAS}\ }\textbf {\bibinfo {volume} {112}},\ \bibinfo
  {pages} {3275} (\bibinfo {year} {2015})}\BibitemShut {NoStop}%
\bibitem [{\citenamefont {Ng}\ \emph {et~al.}(2015)\citenamefont {Ng},
  \citenamefont {Man{\v{c}}inska}, \citenamefont {Cirstoiu}, \citenamefont
  {Eisert},\ and\ \citenamefont {Wehner}}]{Ng2015}%
  \BibitemOpen
  \bibfield  {author} {\bibinfo {author} {\bibfnamefont {N.~H.~Y.}\
  \bibnamefont {Ng}}, \bibinfo {author} {\bibfnamefont {L.}~\bibnamefont
  {Man{\v{c}}inska}}, \bibinfo {author} {\bibfnamefont {C.}~\bibnamefont
  {Cirstoiu}}, \bibinfo {author} {\bibfnamefont {J.}~\bibnamefont {Eisert}}, \
  and\ \bibinfo {author} {\bibfnamefont {S.}~\bibnamefont {Wehner}},\ }\href
  {\doibase 10.1088/1367-2630/17/8/085004} {\bibfield  {journal} {\bibinfo
  {journal} {New J. Phys.}\ }\textbf {\bibinfo {volume} {17}},\ \bibinfo
  {pages} {085004} (\bibinfo {year} {2015})}\BibitemShut {NoStop}%
\bibitem [{\citenamefont {M{\"u}ller}(2018)}]{Mueller2017}%
  \BibitemOpen
  \bibfield  {author} {\bibinfo {author} {\bibfnamefont {M.~P.}\ \bibnamefont
  {M{\"u}ller}},\ }\href {\doibase 10.1103/physrevx.8.041051} {\bibfield
  {journal} {\bibinfo  {journal} {Phys. Rev. X}\ }\textbf {\bibinfo {volume}
  {8}} (\bibinfo {year} {2018}),\ 10.1103/physrevx.8.041051}\BibitemShut
  {NoStop}%
\bibitem [{\citenamefont {Boes}\ \emph {et~al.}(2019)\citenamefont {Boes},
  \citenamefont {Eisert}, \citenamefont {Gallego}, \citenamefont {Mueller},\
  and\ \citenamefont {Wilming}}]{Boes2018a}%
  \BibitemOpen
  \bibfield  {author} {\bibinfo {author} {\bibfnamefont {P.}~\bibnamefont
  {Boes}}, \bibinfo {author} {\bibfnamefont {J.}~\bibnamefont {Eisert}},
  \bibinfo {author} {\bibfnamefont {R.}~\bibnamefont {Gallego}}, \bibinfo
  {author} {\bibfnamefont {M.~P.}\ \bibnamefont {Mueller}}, \ and\ \bibinfo
  {author} {\bibfnamefont {H.}~\bibnamefont {Wilming}},\ }\href {\doibase
  10.1103/physrevlett.122.210402} {\bibfield  {journal} {\bibinfo  {journal}
  {Phys. Rev. Lett.}\ }\textbf {\bibinfo {volume} {122}},\ \bibinfo {pages}
  {210402} (\bibinfo {year} {2019})}\BibitemShut {NoStop}%
\bibitem [{\citenamefont {Jarzynski}(2011)}]{jarzynski2011equalities}%
  \BibitemOpen
  \bibfield  {author} {\bibinfo {author} {\bibfnamefont {C.}~\bibnamefont
  {Jarzynski}},\ }\href {\doibase 10.1007/978-3-0348-0359-5_4} {\bibfield
  {journal} {\bibinfo  {journal} {Annu. Rev. Condens. Matter Phys.}\ }\textbf
  {\bibinfo {volume} {2}},\ \bibinfo {pages} {329} (\bibinfo {year}
  {2011})}\BibitemShut {NoStop}%
\bibitem [{\citenamefont {Cavina}\ \emph {et~al.}(2016)\citenamefont {Cavina},
  \citenamefont {Mari},\ and\ \citenamefont {Giovannetti}}]{cavina2016optimal}%
  \BibitemOpen
  \bibfield  {author} {\bibinfo {author} {\bibfnamefont {V.}~\bibnamefont
  {Cavina}}, \bibinfo {author} {\bibfnamefont {A.}~\bibnamefont {Mari}}, \ and\
  \bibinfo {author} {\bibfnamefont {V.}~\bibnamefont {Giovannetti}},\ }\href
  {\doibase 10.1038/srep29282} {\bibfield  {journal} {\bibinfo  {journal}
  {Scientific Rep.}\ }\textbf {\bibinfo {volume} {6}},\ \bibinfo {pages}
  {29282} (\bibinfo {year} {2016})}\BibitemShut {NoStop}%
\bibitem [{\citenamefont {Maillet}\ \emph {et~al.}(2019)\citenamefont {Maillet}
  \emph {et~al.}}]{maillet2019optimal}%
  \BibitemOpen
  \bibfield  {author} {\bibinfo {author} {\bibfnamefont {O.}~\bibnamefont
  {Maillet}} \emph {et~al.},\ }\href {\doibase 10.1103/PhysRevLett.122.150604}
  {\bibfield  {journal} {\bibinfo  {journal} {Phys. Rev. Lett.}\ }\textbf
  {\bibinfo {volume} {122}},\ \bibinfo {pages} {150604} (\bibinfo {year}
  {2019})}\BibitemShut {NoStop}%
\bibitem [{\citenamefont {Rastegin}(2013)}]{Rastegin_2013}%
  \BibitemOpen
  \bibfield  {author} {\bibinfo {author} {\bibfnamefont {A.~E.}\ \bibnamefont
  {Rastegin}},\ }\href {\doibase 10.1088/1742-5468/2013/06/P06016} {\bibfield
  {journal} {\bibinfo  {journal} {J. Stat. Mech.}\ }\textbf {\bibinfo {volume}
  {2013}},\ \bibinfo {pages} {P06016} (\bibinfo {year} {2013})}\BibitemShut
  {NoStop}%
\bibitem [{\citenamefont {Rastegin}\ and\ \citenamefont {\ifmmode~\dot{Z}\else
  \.{Z}\fi{}yczkowski}(2014)}]{Rastegin_2014}%
  \BibitemOpen
  \bibfield  {author} {\bibinfo {author} {\bibfnamefont {A.~E.}\ \bibnamefont
  {Rastegin}}\ and\ \bibinfo {author} {\bibfnamefont {K.}~\bibnamefont
  {\ifmmode~\dot{Z}\else \.{Z}\fi{}yczkowski}},\ }\href {\doibase
  10.1103/PhysRevE.89.012127} {\bibfield  {journal} {\bibinfo  {journal} {Phys.
  Rev. E}\ }\textbf {\bibinfo {volume} {89}},\ \bibinfo {pages} {012127}
  (\bibinfo {year} {2014})}\BibitemShut {NoStop}%
\bibitem [{\citenamefont {Wilming}\ \emph {et~al.}(2017)\citenamefont
  {Wilming}, \citenamefont {Gallego},\ and\ \citenamefont
  {Eisert}}]{Wilming2017a}%
  \BibitemOpen
  \bibfield  {author} {\bibinfo {author} {\bibfnamefont {H.}~\bibnamefont
  {Wilming}}, \bibinfo {author} {\bibfnamefont {R.}~\bibnamefont {Gallego}}, \
  and\ \bibinfo {author} {\bibfnamefont {J.}~\bibnamefont {Eisert}},\ }\href
  {\doibase 10.3390/e19060241} {\bibfield  {journal} {\bibinfo  {journal}
  {Entropy}\ }\textbf {\bibinfo {volume} {19}},\ \bibinfo {pages} {241}
  (\bibinfo {year} {2017})}\BibitemShut {NoStop}%
\bibitem [{\citenamefont {Boes}\ \emph {et~al.}(2018)\citenamefont {Boes},
  \citenamefont {Wilming}, \citenamefont {Gallego},\ and\ \citenamefont
  {Eisert}}]{Boes_2018}%
  \BibitemOpen
  \bibfield  {author} {\bibinfo {author} {\bibfnamefont {P.}~\bibnamefont
  {Boes}}, \bibinfo {author} {\bibfnamefont {H.}~\bibnamefont {Wilming}},
  \bibinfo {author} {\bibfnamefont {R.}~\bibnamefont {Gallego}}, \ and\
  \bibinfo {author} {\bibfnamefont {J.}~\bibnamefont {Eisert}},\ }\href
  {\doibase 10.1103/PhysRevX.8.041016} {\bibfield  {journal} {\bibinfo
  {journal} {Phys. Rev. X}\ }\textbf {\bibinfo {volume} {8}},\ \bibinfo {pages}
  {041016} (\bibinfo {year} {2018})}\BibitemShut {NoStop}%
\bibitem [{\citenamefont {Alhambra}\ \emph {et~al.}(2019)\citenamefont
  {Alhambra}, \citenamefont {Lostaglio},\ and\ \citenamefont
  {Perry}}]{Alhambra_2019}%
  \BibitemOpen
  \bibfield  {author} {\bibinfo {author} {\bibfnamefont {{\'A}.~M.}\
  \bibnamefont {Alhambra}}, \bibinfo {author} {\bibfnamefont {M.}~\bibnamefont
  {Lostaglio}}, \ and\ \bibinfo {author} {\bibfnamefont {C.}~\bibnamefont
  {Perry}},\ }\href {\doibase 10.22331/q-2019-09-23-188} {\bibfield  {journal}
  {\bibinfo  {journal} {Quantum}\ }\textbf {\bibinfo {volume} {3}},\ \bibinfo
  {pages} {188} (\bibinfo {year} {2019})}\BibitemShut {NoStop}%
\bibitem [{\citenamefont {Sparaciari}\ \emph {et~al.}(2017)\citenamefont
  {Sparaciari}, \citenamefont {Jennings},\ and\ \citenamefont
  {Oppenheim}}]{Sparaciari2017}%
  \BibitemOpen
  \bibfield  {author} {\bibinfo {author} {\bibfnamefont {C.}~\bibnamefont
  {Sparaciari}}, \bibinfo {author} {\bibfnamefont {D.}~\bibnamefont
  {Jennings}}, \ and\ \bibinfo {author} {\bibfnamefont {J.}~\bibnamefont
  {Oppenheim}},\ }\href {\doibase 10.1038/s41467-017-01505-4} {\bibfield
  {journal} {\bibinfo  {journal} {Nat. Commun.}\ }\textbf {\bibinfo {volume}
  {8}},\ \bibinfo {pages} {1895} (\bibinfo {year} {2017})}\BibitemShut
  {NoStop}%
\bibitem [{\citenamefont {Audenaert}\ and\ \citenamefont
  {Scheel}(2008)}]{audenaert2008random}%
  \BibitemOpen
  \bibfield  {author} {\bibinfo {author} {\bibfnamefont {K.~M.~R.}\
  \bibnamefont {Audenaert}}\ and\ \bibinfo {author} {\bibfnamefont
  {S.}~\bibnamefont {Scheel}},\ }\href {\doibase 10.1088/1367-2630/10/2/023011}
  {\bibfield  {journal} {\bibinfo  {journal} {New J. Phys.}\ }\textbf {\bibinfo
  {volume} {10}},\ \bibinfo {pages} {023011} (\bibinfo {year}
  {2008})}\BibitemShut {NoStop}%
\bibitem [{\citenamefont {Lenard}(1978)}]{Lenard1978}%
  \BibitemOpen
  \bibfield  {author} {\bibinfo {author} {\bibfnamefont {A.}~\bibnamefont
  {Lenard}},\ }\href {\doibase 10.1007/BF01011769} {\bibfield  {journal}
  {\bibinfo  {journal} {J. Stat. Phys.}\ }\textbf {\bibinfo {volume} {19}},\
  \bibinfo {pages} {575} (\bibinfo {year} {1978})}\BibitemShut {NoStop}%
\bibitem [{\citenamefont {Pusz}\ and\ \citenamefont
  {Woronowicz}(1978)}]{pusz1978}%
  \BibitemOpen
  \bibfield  {author} {\bibinfo {author} {\bibfnamefont {W.}~\bibnamefont
  {Pusz}}\ and\ \bibinfo {author} {\bibfnamefont {S.~L.}\ \bibnamefont
  {Woronowicz}},\ }\href {\doibase 10.1007/BF01614224} {\bibfield  {journal}
  {\bibinfo  {journal} {Comm. Math. Phys.}\ }\textbf {\bibinfo {volume} {58}},\
  \bibinfo {pages} {273} (\bibinfo {year} {1978})}\BibitemShut {NoStop}%
\bibitem [{\citenamefont {Huang}(1987)}]{huang1987statistical}%
  \BibitemOpen
  \bibfield  {author} {\bibinfo {author} {\bibfnamefont {K.}~\bibnamefont
  {Huang}},\ }\href {\doibase 10.1063/1.3047170} {\emph {\bibinfo {title}
  {Statistical mechanics}}}\ (\bibinfo  {publisher} {Wiley},\ \bibinfo {year}
  {1987})\BibitemShut {NoStop}%
\bibitem [{\citenamefont {Talkner}\ \emph {et~al.}(2008)\citenamefont
  {Talkner}, \citenamefont {H\"anggi},\ and\ \citenamefont
  {Morillo}}]{PhysRevE.77.051131}%
  \BibitemOpen
  \bibfield  {author} {\bibinfo {author} {\bibfnamefont {P.}~\bibnamefont
  {Talkner}}, \bibinfo {author} {\bibfnamefont {P.}~\bibnamefont {H\"anggi}}, \
  and\ \bibinfo {author} {\bibfnamefont {M.}~\bibnamefont {Morillo}},\ }\href
  {\doibase 10.1103/PhysRevE.77.051131} {\bibfield  {journal} {\bibinfo
  {journal} {Phys. Rev. E}\ }\textbf {\bibinfo {volume} {77}},\ \bibinfo
  {pages} {051131} (\bibinfo {year} {2008})}\BibitemShut {NoStop}%
\bibitem [{\citenamefont {Lostaglio}\ \emph {et~al.}(2015)\citenamefont
  {Lostaglio}, \citenamefont {M{\"u}ller},\ and\ \citenamefont
  {Pastena}}]{Pastena}%
  \BibitemOpen
  \bibfield  {author} {\bibinfo {author} {\bibfnamefont {M.}~\bibnamefont
  {Lostaglio}}, \bibinfo {author} {\bibfnamefont {M.~P.}\ \bibnamefont
  {M{\"u}ller}}, \ and\ \bibinfo {author} {\bibfnamefont {M.}~\bibnamefont
  {Pastena}},\ }\href {\doibase 10.1103/PhysRevLett.115.150402} {\bibfield
  {journal} {\bibinfo  {journal} {Phys. Rev. Lett.}\ }\textbf {\bibinfo
  {volume} {115}},\ \bibinfo {pages} {150402} (\bibinfo {year}
  {2015})}\BibitemShut {NoStop}%
\bibitem [{\citenamefont {Sagawa}\ and\ \citenamefont
  {Ueda}(2012{\natexlab{a}})}]{sagawa2012fluctuation}%
  \BibitemOpen
  \bibfield  {author} {\bibinfo {author} {\bibfnamefont {T.}~\bibnamefont
  {Sagawa}}\ and\ \bibinfo {author} {\bibfnamefont {M.}~\bibnamefont {Ueda}},\
  }\href {\doibase 10.1103/PhysRevLett.109.180602} {\bibfield  {journal}
  {\bibinfo  {journal} {Phys. Rev. Lett.}\ }\textbf {\bibinfo {volume} {109}},\
  \bibinfo {pages} {180602} (\bibinfo {year} {2012}{\natexlab{a}})}\BibitemShut
  {NoStop}%
\bibitem [{\citenamefont {Sagawa}\ and\ \citenamefont
  {Ueda}(2013)}]{sagawa2013role}%
  \BibitemOpen
  \bibfield  {author} {\bibinfo {author} {\bibfnamefont {T.}~\bibnamefont
  {Sagawa}}\ and\ \bibinfo {author} {\bibfnamefont {M.}~\bibnamefont {Ueda}},\
  }\href {\doibase 10.1088/1367-2630/15/12/125012} {\bibfield  {journal}
  {\bibinfo  {journal} {New J. Phys.}\ }\textbf {\bibinfo {volume} {15}},\
  \bibinfo {pages} {125012} (\bibinfo {year} {2013})}\BibitemShut {NoStop}%
\bibitem [{\citenamefont {Sagawa}\ and\ \citenamefont
  {Ueda}(2012{\natexlab{b}})}]{PhysRevLett.109.180602}%
  \BibitemOpen
  \bibfield  {author} {\bibinfo {author} {\bibfnamefont {T.}~\bibnamefont
  {Sagawa}}\ and\ \bibinfo {author} {\bibfnamefont {M.}~\bibnamefont {Ueda}},\
  }\href {\doibase 10.1103/PhysRevLett.109.180602} {\bibfield  {journal}
  {\bibinfo  {journal} {Phys. Rev. Lett.}\ }\textbf {\bibinfo {volume} {109}},\
  \bibinfo {pages} {180602} (\bibinfo {year} {2012}{\natexlab{b}})}\BibitemShut
  {NoStop}%
\bibitem [{\citenamefont {Toyabe}\ \emph {et~al.}(2010)\citenamefont {Toyabe},
  \citenamefont {Sagawa}, \citenamefont {Ueda}, \citenamefont {Muneyuki},\ and\
  \citenamefont {Sano}}]{toyabe2010experimental}%
  \BibitemOpen
  \bibfield  {author} {\bibinfo {author} {\bibfnamefont {S.}~\bibnamefont
  {Toyabe}}, \bibinfo {author} {\bibfnamefont {T.}~\bibnamefont {Sagawa}},
  \bibinfo {author} {\bibfnamefont {M.}~\bibnamefont {Ueda}}, \bibinfo {author}
  {\bibfnamefont {E.}~\bibnamefont {Muneyuki}}, \ and\ \bibinfo {author}
  {\bibfnamefont {M.}~\bibnamefont {Sano}},\ }\href {\doibase
  10.1038/nphys1821} {\bibfield  {journal} {\bibinfo  {journal} {Nature Phys.}\
  }\textbf {\bibinfo {volume} {6}},\ \bibinfo {pages} {988} (\bibinfo {year}
  {2010})}\BibitemShut {NoStop}%
\bibitem [{\citenamefont {Horodecki}\ and\ \citenamefont
  {Oppenheim}(2013)}]{horodecki2013fundamental}%
  \BibitemOpen
  \bibfield  {author} {\bibinfo {author} {\bibfnamefont {M.}~\bibnamefont
  {Horodecki}}\ and\ \bibinfo {author} {\bibfnamefont {J.}~\bibnamefont
  {Oppenheim}},\ }\href {\doibase 10.1038/ncomms3059} {\bibfield  {journal}
  {\bibinfo  {journal} {Nature Comm.}\ }\textbf {\bibinfo {volume} {4}},\
  \bibinfo {pages} {2059} (\bibinfo {year} {2013})}\BibitemShut {NoStop}%
\bibitem [{\citenamefont {Perry}\ \emph {et~al.}(2018)\citenamefont {Perry},
  \citenamefont {\ifmmode \acute{C}\else \'{C}\fi{}wikli\ifmmode~\acute{n}\else
  \'{n}\fi{}ski}, \citenamefont {Anders}, \citenamefont {Horodecki},\ and\
  \citenamefont {Oppenheim}}]{Perry_2018}%
  \BibitemOpen
  \bibfield  {author} {\bibinfo {author} {\bibfnamefont {C.}~\bibnamefont
  {Perry}}, \bibinfo {author} {\bibfnamefont {P.}~\bibnamefont {\ifmmode
  \acute{C}\else \'{C}\fi{}wikli\ifmmode~\acute{n}\else \'{n}\fi{}ski}},
  \bibinfo {author} {\bibfnamefont {J.}~\bibnamefont {Anders}}, \bibinfo
  {author} {\bibfnamefont {M.}~\bibnamefont {Horodecki}}, \ and\ \bibinfo
  {author} {\bibfnamefont {J.}~\bibnamefont {Oppenheim}},\ }\href {\doibase
  10.1103/PhysRevX.8.041049} {\bibfield  {journal} {\bibinfo  {journal} {Phys.
  Rev. X}\ }\textbf {\bibinfo {volume} {8}},\ \bibinfo {pages} {041049}
  (\bibinfo {year} {2018})}\BibitemShut {NoStop}%
\bibitem [{\citenamefont {Faist}\ \emph {et~al.}(2015)\citenamefont {Faist},
  \citenamefont {Oppenheim},\ and\ \citenamefont {Renner}}]{faist2015gibbs}%
  \BibitemOpen
  \bibfield  {author} {\bibinfo {author} {\bibfnamefont {P.}~\bibnamefont
  {Faist}}, \bibinfo {author} {\bibfnamefont {J.}~\bibnamefont {Oppenheim}}, \
  and\ \bibinfo {author} {\bibfnamefont {R.}~\bibnamefont {Renner}},\ }\href
  {\doibase 10.1088/1367-2630/17/4/043003} {\bibfield  {journal} {\bibinfo
  {journal} {New J. Phys.}\ }\textbf {\bibinfo {volume} {17}},\ \bibinfo
  {pages} {043003} (\bibinfo {year} {2015})}\BibitemShut {NoStop}%
\bibitem [{\citenamefont {Gogolin}\ \emph {et~al.}(2011)\citenamefont
  {Gogolin}, \citenamefont {M\"uller},\ and\ \citenamefont
  {Eisert}}]{PhysRevLett.106.040401}%
  \BibitemOpen
  \bibfield  {author} {\bibinfo {author} {\bibfnamefont {C.}~\bibnamefont
  {Gogolin}}, \bibinfo {author} {\bibfnamefont {M.~P.}\ \bibnamefont
  {M\"uller}}, \ and\ \bibinfo {author} {\bibfnamefont {J.}~\bibnamefont
  {Eisert}},\ }\href {\doibase 10.1103/PhysRevLett.106.040401} {\bibfield
  {journal} {\bibinfo  {journal} {Phys. Rev. Lett.}\ }\textbf {\bibinfo
  {volume} {106}},\ \bibinfo {pages} {040401} (\bibinfo {year}
  {2011})}\BibitemShut {NoStop}%
\bibitem [{\citenamefont {Anshu}(2016)}]{Anshu2016}%
  \BibitemOpen
  \bibfield  {author} {\bibinfo {author} {\bibfnamefont {A.}~\bibnamefont
  {Anshu}},\ }\href {\doibase 10.1088/1367-2630/18/8/083011} {\bibfield
  {journal} {\bibinfo  {journal} {New J. Phys.}\ }\textbf {\bibinfo {volume}
  {18}},\ \bibinfo {pages} {083011} (\bibinfo {year} {2016})}\BibitemShut
  {NoStop}%
\bibitem [{\citenamefont {Nielsen}\ and\ \citenamefont
  {Chuang}(2000)}]{nielsenchuang}%
  \BibitemOpen
  \bibfield  {author} {\bibinfo {author} {\bibfnamefont {M.~A.}\ \bibnamefont
  {Nielsen}}\ and\ \bibinfo {author} {\bibfnamefont {I.~L.}\ \bibnamefont
  {Chuang}},\ }\href {\doibase 10.1119/1.1463744} {\emph {\bibinfo {title}
  {{Quantum Computation and Quantum Information}}}}\ (\bibinfo  {publisher}
  {Cambridge University Press},\ \bibinfo {year} {2000})\BibitemShut {NoStop}%
\bibitem [{\citenamefont {{Goldstein}}\ \emph {et~al.}(2013)\citenamefont
  {{Goldstein}}, \citenamefont {Hara},\ and\ \citenamefont
  {Tasaki}}]{Goldstein_2013}%
  \BibitemOpen
  \bibfield  {author} {\bibinfo {author} {\bibfnamefont {S.}~\bibnamefont
  {{Goldstein}}}, \bibinfo {author} {\bibfnamefont {T.}~\bibnamefont {Hara}}, \
  and\ \bibinfo {author} {\bibfnamefont {H.}~\bibnamefont {Tasaki}},\
  }\href@noop {} {\bibfield  {journal} {\bibinfo  {journal} {ArXiv e-prints}\ }
  (\bibinfo {year} {2013})},\ \Eprint {http://arxiv.org/abs/1303.6393}
  {arXiv:1303.6393} \BibitemShut {NoStop}%
\bibitem [{\citenamefont {Kliesch}\ \emph {et~al.}(2014)\citenamefont
  {Kliesch}, \citenamefont {Gogolin}, \citenamefont {Kastoryano}, \citenamefont
  {Riera},\ and\ \citenamefont {Eisert}}]{Kliesch2014}%
  \BibitemOpen
  \bibfield  {author} {\bibinfo {author} {\bibfnamefont {M.}~\bibnamefont
  {Kliesch}}, \bibinfo {author} {\bibfnamefont {C.}~\bibnamefont {Gogolin}},
  \bibinfo {author} {\bibfnamefont {M.~J.}\ \bibnamefont {Kastoryano}},
  \bibinfo {author} {\bibfnamefont {A.}~\bibnamefont {Riera}}, \ and\ \bibinfo
  {author} {\bibfnamefont {J.}~\bibnamefont {Eisert}},\ }\href {\doibase
  10.1103/PhysRevX.4.031019} {\bibfield  {journal} {\bibinfo  {journal} {Phys.
  Rev. X}\ }\textbf {\bibinfo {volume} {4}},\ \bibinfo {pages} {031019}
  (\bibinfo {year} {2014})}\BibitemShut {NoStop}%
\bibitem [{\citenamefont {Watrous}(2018)}]{watrous2018theory}%
  \BibitemOpen
  \bibfield  {author} {\bibinfo {author} {\bibfnamefont {J.}~\bibnamefont
  {Watrous}},\ }\href {\doibase 10.1017/9781316848142} {\emph {\bibinfo {title}
  {{The theory of quantum information}}}}\ (\bibinfo  {publisher} {Cambridge
  University Press},\ \bibinfo {year} {2018})\BibitemShut {NoStop}%
\bibitem [{\citenamefont {Perarnau-Llobet}\ \emph {et~al.}(2017)\citenamefont
  {Perarnau-Llobet}, \citenamefont {B{\"a}umer}, \citenamefont {Hovhannisyan},
  \citenamefont {Huber},\ and\ \citenamefont {Acin}}]{perarnau2017no}%
  \BibitemOpen
  \bibfield  {author} {\bibinfo {author} {\bibfnamefont {M.}~\bibnamefont
  {Perarnau-Llobet}}, \bibinfo {author} {\bibfnamefont {E.}~\bibnamefont
  {B{\"a}umer}}, \bibinfo {author} {\bibfnamefont {K.~V.}\ \bibnamefont
  {Hovhannisyan}}, \bibinfo {author} {\bibfnamefont {M.}~\bibnamefont {Huber}},
  \ and\ \bibinfo {author} {\bibfnamefont {A.}~\bibnamefont {Acin}},\ }\href
  {\doibase 10.1103/PhysRevLett.118.070601} {\bibfield  {journal} {\bibinfo
  {journal} {Phys. Rev. Lett.}\ }\textbf {\bibinfo {volume} {118}},\ \bibinfo
  {pages} {070601} (\bibinfo {year} {2017})}\BibitemShut {NoStop}%
\bibitem [{\citenamefont {Koski}\ \emph {et~al.}(2014)\citenamefont {Koski},
  \citenamefont {Maisi}, \citenamefont {Sagawa},\ and\ \citenamefont
  {Pekola}}]{koski2014experimental}%
  \BibitemOpen
  \bibfield  {author} {\bibinfo {author} {\bibfnamefont {J.~V.}\ \bibnamefont
  {Koski}}, \bibinfo {author} {\bibfnamefont {V.~F.}\ \bibnamefont {Maisi}},
  \bibinfo {author} {\bibfnamefont {T.}~\bibnamefont {Sagawa}}, \ and\ \bibinfo
  {author} {\bibfnamefont {J.~P.}\ \bibnamefont {Pekola}},\ }\href {\doibase
  10.1103/PhysRevLett.113.030601} {\bibfield  {journal} {\bibinfo  {journal}
  {Phys. Rev. Lett.}\ }\textbf {\bibinfo {volume} {113}},\ \bibinfo {pages}
  {030601} (\bibinfo {year} {2014})}\BibitemShut {NoStop}%
\end{thebibliography}%
\newpage\appendix
\renewcommand{\thesection}{\Alph{section}}
\section{NMW for Gibbs preserving maps}
\label{app:nmw_for_gp}
Thermalizing quantum maps, in particular those studied in the resource theoretic framework, are maps that model the evolution of a non-equilibrium quantum state as it exchanges heat with its surrounding thermal bath. Several variants of these maps exist~\cite{horodecki2013fundamental,Brandao2015,Ng2015,Perry_2018,faist2015gibbs}, but a common feature is that they are \emph{Gibbs preserving (GP)}, namely that the Gibbs canonical state is a fixed point of such maps. Thermalizing maps are often viewed as ``free operations'' in a resource theoretic context, since they allow only for heat (instead of work) exchange with an environment in thermal equilibrium.
In this section, we demonstrate two things: First, that even such thermodynamically ``cheap'' channels may violate the JE very strongly, due to non-unitality. Secondly, that they cannot be used to violate the NMW condition. A diagrammatic overview over the various properties of channels with respect to JE and NMW is given in Fig.~\ref{fig:diagram_violation}. 

We now turn to the first point. Given a $ d $-dimensional system $ S $ with Hamiltonian $ H $, the violation of JE can be calculated for the thermalizing channel as
\begin{align}
	\langle e^{-\beta W} \rangle & = d   \sum_{j} \frac{e^{-\beta E_j}}{Z_H} \bra{E_j} \mc{C} \left[\one/d\right] \ket{E_j}, \\
	                             & =d                                                                                        
	\sum_{j} \frac{e^{-\beta E_j}}{Z_H} \bra{E_j} \omega_\beta (H) \ket{E_j}=\frac{d}{d_{\rm eff}},
\end{align}
where $ d_{\rm eff}:={1}/{\tr (\omega_\beta(H)^2)} $ is known as the effective dimension \cite{PhysRevLett.106.040401} of the thermal state. One sees from the above that JE is always violated for $ \beta >0 $, since $ d_{\rm eff} \leq d $, with equality only when $ \omega_\beta(H)=\one/d $ is maximally mixed. For $ N $ non-interacting i.i.d.\ systems, both $ d $ and $ \tr(\rho^2) $ scale exponentially with $ N $, leading to an exponential violation in $ N $ for JE.

Turning to the second point, one may wonder how this notion of thermodynamically free channels can be reconciled with the fact that JE is violated. However, note that in the standard JE setting, the work variable is traditionally defined in terms of a fluctuating (measured) energy difference in the system, and does not inherently distinguish between work and heat contributions -- unlike resource-theoretic settings where heat flow is allowed for free, but measurements incur a thermodynamic cost. Here, we consider an operationally more meaningful characterization (NMW as defined in Def.~1 of the main text), and show that NMW cannot be violated using channels that preserve the Gibbs state in generic many-body systems. 
The only assumptions that we make are that i) the system has uniformly bounded, local interactions on a $D$-dimensional regular lattice and ii) a finite correlation length, i.e., the temperature is non-critical.
\begin{lemma}[Non-violation of NMW for Gibbs-preserving maps]\label{lemma:NMW-GP}
	No channel $\mc E$ that preserves the Gibbs state can violate NMW for locally interacting many-body systems at a non-critical temperature.
\end{lemma}
\begin{proof}
	We aim at showing that for any $a>0$, $p(w\geq a) = p(W\geq aN) \rightarrow 0$ as $N\rightarrow \infty$.
	The basic idea behind our proof is to make use of typicality. 
	Let $e^{(N)}$ denote the energy density of the $N$-particle system and denote by $\Pi_{\delta}^{(N)}$ the projector onto energy eigenstates with energies in the interval $T_{N,\delta}:=[(e^{(N)} - \delta)N,(e^{(N)}+\delta)N]$. 
	Finally, denote by $p(\cdot)$ the initial probability distribution of energy of the thermal state $\tau_S^{(N)}$, e.g., the probability that the initial energy measurement yields $E_i\in T_{N,\delta}$ is given by
	\begin{align}
		p(T_{N,\delta}) := \tr\left(\tau_S^{(N)} \Pi_\delta^{(N)}\right). 
	\end{align}
		
	A theorem by Anshu \cite{Anshu2016} shows that under the given conditions most weight of the thermal state $\tau_S^{(N)}$ of the $N$-particle system is contained in a typical subspace. 
	More precisely, for a many-body system described by a $D$-dimensional lattice, there exist constants $C,K>0$ such that for any $\delta>0$ we have  
	\begin{align}\label{anshu}
		p(T_{N,\delta}) \geq 1 - C \e^{-\frac{(\delta^2 N)^\frac{1}{1+D}}{K}}. 
	\end{align}
	This is equivalent to saying that 
	\begin{equation}
		p(T_{N,\delta}^c) \leq C \e^{-\frac{(\delta^2 N)^\frac{1}{1+D}}{K}}, 
	\end{equation}
	where $T_{N,\delta}^c = \RR\setminus T_{N,\delta}$. 
	In particular, in the case of $D=0$, i.e., $N$ non-interacting systems, we find the usual scaling obtained from Hoeffding's inequality.
	In the following, for simplicity of notation, we write $\sigma_1=\tau_S^{(N)}$ and consider the normalized state $\sigma_2$ obtained by restricting $\tau_S^{(N)}$ to the subspace $\Pi_\delta^{(N)}$ as
	\begin{align}
		\sigma_2 := \frac{\Pi_\delta^{(N)} \tau_S^{(N)}}{p(T_{N,\delta})}. 
	\end{align}
	Let us further write $\mc E (\sigma_{1(2)}) = \sigma_{1(2)}'$, where $\sigma_1' = \sigma_1$ by assumption.
	Since the trace distance $d(\rho_1,\rho_2) \coloneqq \frac{1}{2}\tr(|\rho_1 - \rho_2|)$ fulfills the data processing inequality, 	
	\begin{equation}
		d(\sigma_1,\sigma_2') = d(\sigma_1',\sigma_2') \leq d(\sigma_1,\sigma_2) = p(T_{N,\delta}^c).
	\end{equation}
	Using the operational meaning of trace distance $d(\rho_1,\rho_2) = \displaystyle\max_{0\leq M\leq I} |\tr (M(\rho_1-\rho_2))|$ \cite{nielsenchuang}, this means that 
	\begin{equation}\label{key3}
		|\tr(\Pi_{\delta}^{(N)} \sigma_1) - \tr (\Pi_{\delta}^{(N)}\sigma_2') | \leq p(T_{N,\delta}^c)
	\end{equation}
	and, in turn, 
	\begin{align}\label{key2}
		\tr (\Pi_{\delta}^{(N)} \sigma_2') \geq p(T_{N,\delta})-p(T_{N,\delta}^c) = 1-2 p(T_{N,\delta}^c). 
	\end{align}
	To see this, note that~\eqref{key2} follows from~\eqref{key3} directly if $\tr (\Pi_{\delta}^{(N)}\sigma_2') \leq \tr(\Pi_{\delta}^{(N)} \sigma_1)$, and as 
	\begin{equation}
		\tr (\Pi_{\delta}^{(N)}\sigma_2') > \tr(\Pi_{\delta}^{(N)} \sigma_1) \geq \tr(\Pi_{\delta}^{(N)} \sigma_1) - p(T_{N,\delta}^c)
	\end{equation}
	otherwise.
	This means that, conditioned on the fact that the initial state was within the typical energy window ($E_i\in T_{N,\delta}$), the final energy $ E_f $ is also within this energy window except with probability $ 2 p(T_{N,\delta}^c) $, which is (sub-)exponentially small in $ N $. We will use this later.
		
	We are now ready to evaluate the probability of obtaining macroscopic work. 
	\begin{align*}		
		p(w\geq a ) & = p(T_{N,\delta}) \cdot p(w\geq a | E_i\in T_{N,\delta})                \\
		            & \quad\quad+ p(T_{N,\delta}^c)\cdot p (w\geq a | E_i \in T_{N,\delta}^c) \\
		            & \leq p(w\geq a | E_i \in T_{N,\delta}) + p(T_{N,\delta}^c).             
	\end{align*}
	We can estimate the first term as
	\begin{align*}
		p(w\geq a |  E_i \in T_{N,\delta}) & \leq p(E_f \leq (e^{(N)} +\delta - a)N| E_i\in T_{N,\delta}). 
	\end{align*}
	We now choose $\delta = a/2$ and get
	\begin{align*}
		p(w\geq a |  E_i \in T_{N,\delta}) & \leq p(E_f \leq (e^{(N)} - a/2)N | E_i \in T_{N,\delta})           \\
		                                   & \leq \tr\left[\sigma_2' \left(\one - \Pi_{a/2}^{(N)}\right)\right] \\
		                                   & \leq 2 p(T_{N,a/2}^c),                                             
	\end{align*}
	where we have used \eqref{key2} in the last step.
	Altogether, we thus find
	\begin{align*}
		p (w\geq a ) & \leq 3 p(T_{N,a/2}^c), 
	\end{align*}
	which decays to zero (sub-)exponentially by \eqref{anshu}. This concludes the proof.
\end{proof}
As a side-remark, we note that if the Gibbs-preserving channels that appear here are interpreted as modelling the interaction with a heat bath, then the above result can be interpreted as a "no macroscopic heat" statement: If a macroscopic system is brought in thermal contact with a heat bath at the same temperature, then the probability of an exchange of a macroscopic amount of heat is arbitrarily small in the system size. 

\begin{figure}[t]
	\centering
	\includegraphics[width=0.33\textwidth]{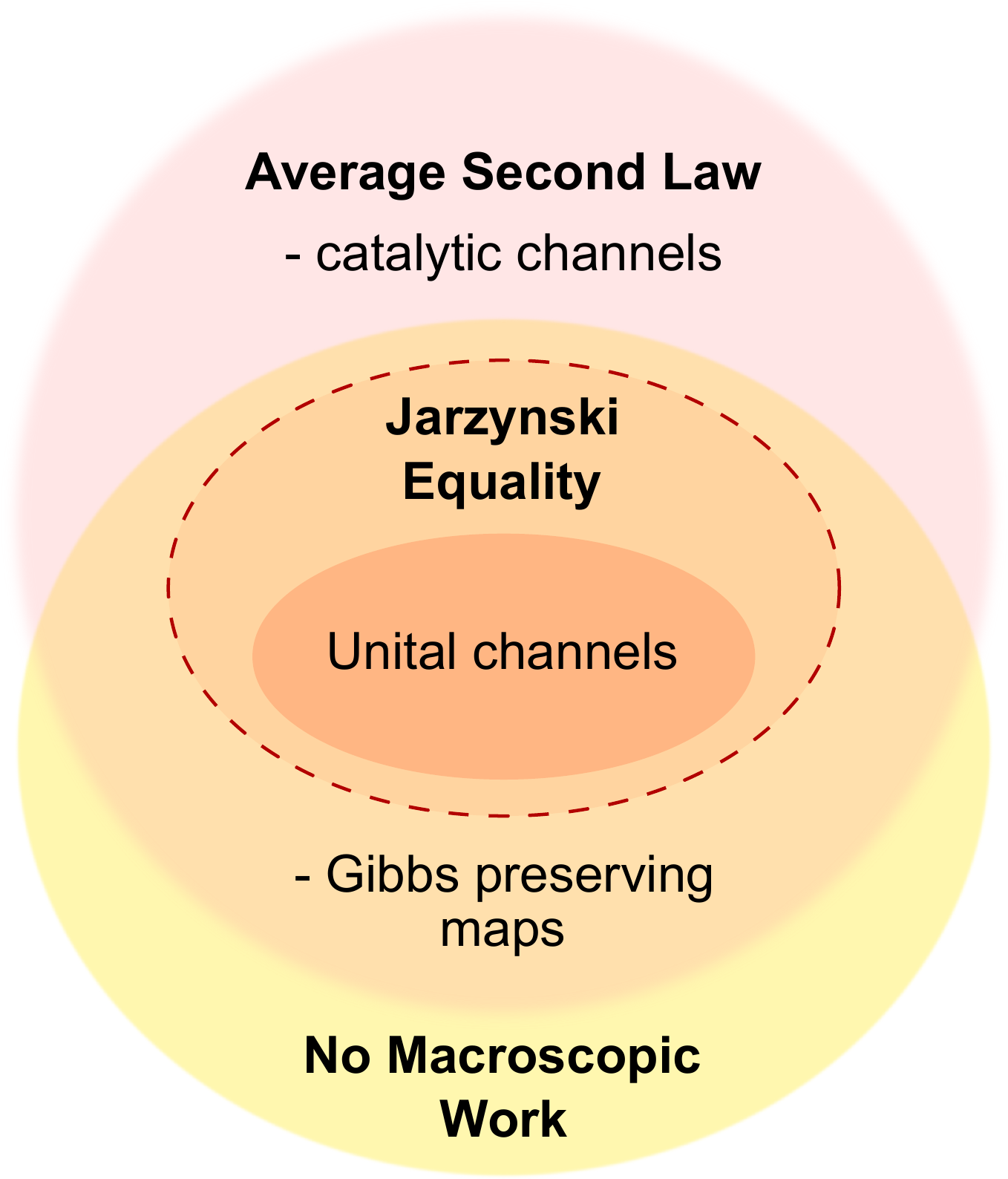}
	\caption{A summary of different criteria (Av-SL, NMW and JE) mentioned in the main text, with examples of maps according to this characterization.}
	\label{fig:diagram_violation}
\end{figure}
\section{Microscopic toy example} 
\label{app:microscopic_toy_example}

In this section, we show that already for small systems and using catalysts, the JE can be violated. We do so by constructing non-unital catalytic channels. Indeed, such maps can be realized ``quasi-classically'', in the sense that in the construction it is sufficient to consider the energy spectra of the involved states and that all unitaries are simple permutations of those values. We consider a 3-level system with energy levels $E_1=0,E_2=0,E_3=\Delta$ in the thermal state 
\begin{align}
	w= \left(\frac{1}{Z},\frac{1}{Z},\frac{Z-2}{Z}\right), 
\end{align}
where $Z= 2+ \exp(-\beta\Delta)$ is the partition function and we express the state as a probability vector, such that $w_i$ denotes the $i$th eigenvalue of the thermal state. For later, we observe that $2\leq Z\leq 3$.

We are going to construct a simple non-unital catalytic channel that involves a 2-dimensional catalyst. 
Let $e_{i}$ and $f_j$ denote the basis states for the vector spaces $\mc V_S$ and $\mc V_C$ describing the system and catalyst respectively. We define the permutation $\pi$ acting on the joint vector space $\mc V_S \otimes \mc V_C$ as that permutation which exchanges the respective levels $e_1 \otimes f_1 \Leftrightarrow e_2 \otimes f_2$ and $e_2 \otimes f_1 \Leftrightarrow e_3 \otimes f_2$ and leaves all other entries unchanged (see Fig.~\ref{fig:example}).
For the catalyst to remain unchanged for this permutation and initial system state, it is easy to check that the catalyst has to be given by the vector
\begin{align}
	q	= \left(\frac{Z-1}{Z+1},\frac{2}{Z+1}\right). 
\end{align}
Now, the catalytic channel $\mc C$ induced by this catalyst and permutation on the system has the general effect 
\begin{align}
	\mc C(p_1,p_2,p_3) = (q_1 p_2 + q_1 p_3, q_1 p_1 + q_2 p_3, q_2 p_1 + q_2 p_2), 
\end{align}
so that, in particular, the maximally mixed input state is mapped to 
\begin{align}
	\mc C(\one/3) = \frac{2}{3}\left(\frac{Z-1}{Z+1}, \frac{1}{2}, \frac{2}{Z+1}\right), 
\end{align}
which is different from the maximally mixed vector for any $\Delta > 0$.

\begin{figure}[t]
	\begin{tabular}{c  c | c}
		$\red{q_2 p_3}$ & $q_1 p_3$        & $p_3$ \\
		$\red{q_2 p_2}$ & $\blue{q_1 p_2}$ & $p_2$ \\
		$q_2 p_1$       & $\blue{q_1 p_1}$ & $p_1$ \\
		\hline
		$q_2$           & $q_1$            &       
	\end{tabular}
	$\: \rightarrow \:$
	\begin{tabular}{ c  c | c}
		$\blue{q_1 p_2}$ & $q_1 p_3$       & $q_1 p_2 + q_1 p_3$ \\
		$\blue{q_1 p_1}$ & $\red{q_2 p_3}$ & $q_1 p_1 + q_2 p_3$ \\
		$q_2 p_1$        & $\red{q_2 p_2}$ & $q_2 p_1 + q_2 p_2$ \\
		\hline
		$q_2$            & $q_1$           &                     
	\end{tabular}
	\caption{We represent the joint state of system and catalyst by means of a table. \emph{Left:} At the beginning the joint system starts out in a product state, so that the entry $(i,j)$ is given by the product of the $i$th eigenvalue of the system and $j$th eigenvalue of the catalyst. \emph{Right:} After applying the permutation highlighted in red, the marginal state of the system, given by the rows sums, has changed, while the marginal state of the catalyst (given by the column sums), has to remain invariant. For a two-dimensional catalyst, specifying the permutation and initial system state fixes the catalyst state.} 
	\label{fig:example}
\end{figure}
What is more, we can also directly calculate the work-distribution $p(w)$, yielding
\begin{align}
	p(0)       & =\frac{1}{Z(Z+1)}\left[Z+3+2(Z-2)(Z-1)\right], \\
	p(\Delta)  & = \frac{2(Z-2)}{Z(Z+1)},                       \\
	p(-\Delta) & = \frac{Z-1}{Z(Z+1)}.                          
\end{align}
We now want to compute $\langle \e^{\beta W}\rangle$. 
To do so, it is useful to note that $e^{-\beta \Delta}= Z-2$ and hence $\e^{\beta\Delta}=1/(Z-2)$. 
We find
\begin{align}
	\langle \e^{\beta W}\rangle = \frac{Z+ 5 + 2(Z-2)(Z-1)}{Z(Z+1)} \geq 1. 
\end{align}
In fact, this quantity is larger than $1$ whenever $Z<3$, corresponding to $\Delta >0$. Its maximum is given as $7/6$ for $Z=3$, which corresponds to $\Delta \rightarrow \infty$. Thus, the Jarzynski inequality is violated. At the same time the second law is fulfilled as expected, since $p(-\Delta) \geq p(\Delta)$. 

\section{Work extraction for initial micro-canonical ensembles} 
\label{app:microcanonical}
In this appendix, we show that a statement similar to Proposition~1 of the main text holds in the slightly different setting of a micro-canonical initial state. 
This serves two purposes: i) in statistical mechanics, one often assumes that closed, macroscopic systems are described by microcanonical ensembles due to the postulate of equal a priori probabilities of microstates corresponding to a macrostate. 
ii) The proof for the microcanonical initial state is conceptually simpler, but also provides the blueprint for the slightly more involved proof in the case of a canonical state, which is provided in Sec.~\ref{app:canonical}.

In the following, we denote by $I\subset \RR$ an energy window, by $g(I)$ the number of energy eigenstates in this window,
\begin{align}
	g(I) = \sum_{E_i \in I}1, 
\end{align}
and the corresponding micro-canonical state by
\begin{align}
	\Omega_{\rm S}(I) = \frac{1}{g(I)}\sum_{E_i \in I} \proj{E_i}. 
\end{align}
A micro-canonical energy window around energy density $e$ is any energy window $I(e)$ of the form $[e- O(\sqrt{N}),e]$, where $N$ is the number of particles.

The only difference to the standard setting described in the main text (as depicted in Fig.~1) is that the initial state differs from the thermal state $\omega_\beta(H)$. Instead, it
is given by the micro-canonical ensemble. In other words, given a micro-canonical energy window $I$, we consider channels $ \mc C $ of the form
\begin{align}
	  & \mathcal{C}(\cdot) = {\rm Tr}_{\rm C} (U (\cdot\otimes\sigma_{\rm C})U^\dagger)                    \\
	  & {\rm s.t.}~ {\rm Tr}_{\rm S} (U (\Omega_{\rm S}(I)\otimes\sigma_{\rm C})U^\dagger)=\sigma_{\rm C}. 
\end{align}
We carry over notation from the main text, so that $p(w \geq \epsilon)$ denotes the probability of measuring the system's energy per particle decrease by at least an amount $\epsilon$, and so on. Furthermore, we take the catalyst Hamiltonian in our construction to be $ H_C = \mathbf{I} $.

We will now first show that the NMW principle also holds for micro-canonical states of generic many-body systems. After that we will show that it can be circumvented using catalysts. 
To show the validity of the NMW principle we will use the same reasoning as presented in Ref.~\cite{Goldstein_2013}, where the NMW principle has been studied before. Thus, the following proof is essentially a reproduction for the convenience of the reader. We consider a sequence of many-body Hamiltonians $H_S^{(N)}$ on $N$ particles with the generic property of having an exponential density of states:
\begin{align}\label{eq:densityofstates}
	g((-\infty,E]) \coloneqq \sum_{E_i \leq E}1 = \e^{N \mu(E/N) - o(N)}, 
\end{align}
where $\mu$ is a strictly monotonic and differentiable function independent of $N$ and $o(N)$ denotes terms small compared to $N$, $\lim_{N\to \infty} o(N)/N = 0$. 
\begin{proposition}[NMW for micro-canonical states] Consider a sequence of $N$-particle Hamiltonians fulfilling \eqref{eq:densityofstates} and a sequence of micro-canonical energy-windows $I^{(N)}=[eN, eN + \delta\sqrt{N}]$ around energy density $e$ (with $\delta>0$ fixed). Then for any unital channel acting on the $N$-particle system, the probability of extracting work $w$ per particle is bounded as
	\begin{align}
		p(w > \epsilon) \leq C \e^{- \mu'(e) \epsilon N + o(N)}, 
	\end{align}
	where $C>0$ is a constant and $\mu'$ denotes the derivative of $\mu$.
	\begin{proof}
		Let $I_\leq := (-\infty,(e-\epsilon)N + \delta \sqrt{N}]$, denote by $P_{\rm S}(I_\leq )$ the projector onto energy-eigenstates with energies below $(e-\epsilon)N +\delta\sqrt{N}$ and let $\mc U$ denote a unital channel. 	In the following, we write $I$ instead of $I^{(N)}$ to simplify notation. Then	
		\begin{align}
			p(w>\epsilon) & \leq \tr\left(P_{\rm S}(I_\leq) \mc U [\Omega_{\rm S}(I)]\right)                                        \\ 
			              & = \sum_{E_i \in I} \frac{1}{g(I)} \tr\left( P_{\rm S}(I_\leq) \mc U\left[ \proj{E_i}\right]\right)      \\
			              & \leq \frac{1}{g(I)} \tr\left(  P_{\rm S}(I_\leq)\mc U\left[\one\right]\right) = \frac{g(I_\leq)}{g(I)}. 
		\end{align}
		Writing $\tilde e := e + \delta N^{-1/2}$, we have		
		\begin{align}
			g(I) & = \e^{N \mu(\tilde e) - o(N)} - \e^{N \mu(e) - o(N)}                             \\
			     & = \e^{N\mu(\tilde e)-o(N)}\left(1 - \e^{-N(\mu(\tilde e)-\mu(e)) + o(N) }\right) \\
			     & \approx \e^{N\mu(\tilde e) - o(N)},                                              
		\end{align}
		where in the last estimation we use that $\mu$ is strictly monotonic. In particular, we can estimate the exponential in the parenthesis as
		\begin{align}
			\e^{-N(\mu(\tilde e) -\mu(e))-o(N)} = O\left( \e^{-\delta \mu'(e)N^{1/2}}\right), 
		\end{align}
		where $\mu'$ denotes the derivative of $\mu$.
		Using $g(I_\leq) = \e^{N(\mu(\tilde e-\epsilon) - o(N)}$ we then find
		\begin{align}
			p(w>\epsilon) & \leq \frac{\e^{- N\left(\mu(\tilde e) - \mu(\tilde e-\epsilon)\right) +o(N)}}{1 - O(\e^{-\delta \mu'(e) \sqrt{N}})} \\
			              & \leq C\e^{- \mu'(e) \epsilon N}.                                                                                    
		\end{align}
	\end{proof}
\end{proposition}
We have here used that $\mu$ is differentiable to prove this result. Similar results would follow for weaker notions of regularity of $\mu$, such as Lipschitz-continuity.
Having proven the NMW principle for generic many-body systems, let us now show how to circumvent it using catalysts.

\begin{proposition}[Overcoming NMW using catalysts]\label{prop:micro} Consider a Hamiltonian $H_S$ and a microcanonical state $\Omega_{\rm S}(I)$, with $I$ a micro-canonical energy window around energy density $e$. 
	{Suppose there exists an energy window $I_+$ with $g(I_+) = g(I)^2$.}
	Then{, for any {$0 \leq  e_- < e$}, }there exists a catalytic channel such that
	\begin{align}\label{eq:resultmicro}
		p(w \geq e-e_-) = \frac{1}{2}. 
	\end{align}
\end{proposition}
Before giving the proof of the proposition, let us emphasize again that the required conditions on the Hamiltonian are very weak.
In particular, the conditions are (approximately) fulfilled if the density of states is well approximated by an exponential in the range of energies that we are working in, a condition that is typically fulfilled in many-body systems and, as we have seen above, leads to an NMW principle if we do not allow for catalysts.
\begin{figure*}[htbp]
	\centering
	\includegraphics[width=0.5\textwidth]{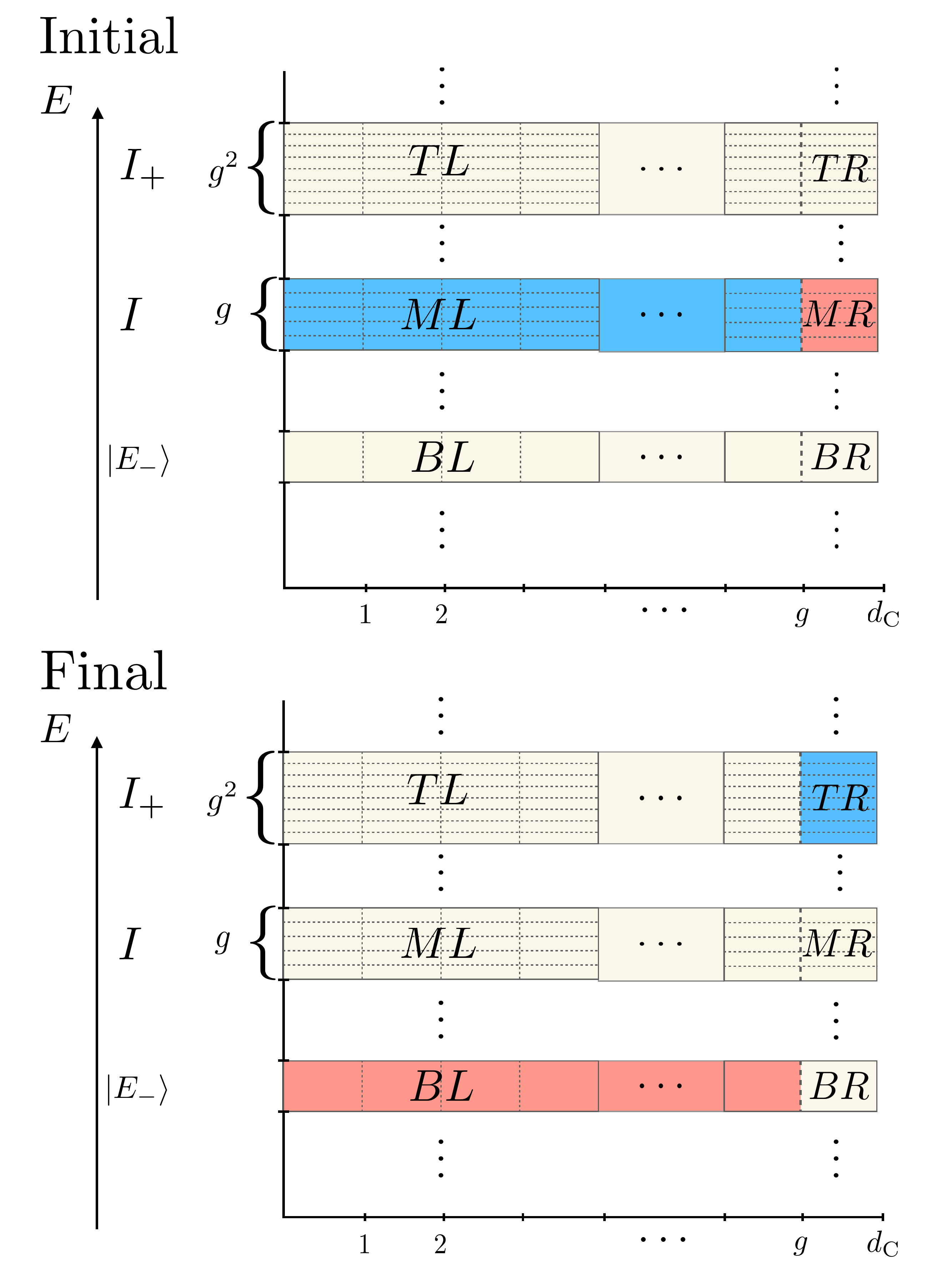}
	\caption{Proof sketch for Proposition~\ref{prop:micro}: \emph{Top:} We represent the initial product state of system and catalyst by means of a table, using the fact that both are initially diagonal in the energy eigenbasis: Ordering the spectra of both states non-increasingly, the entry $(i,j)$ of the table corresponds to the product of the $i$-th eigenvalue of the system (corresponding to the a particular energy eigenstate) and the $j$-th energy eigenvalue of the catalyst. 
		We focus on three regions in the table---denoted top (T), middle (M), bottom (B)---corresponding to {two degeneracy bands $I$ and ${I_+}$ and (the projector onto) a single energy eigenvector $\ket{E_-}$}: Since the system is initially in the micro-canonical ensemble with energy window $I$, the support of the joint state is initially contained in the coloured middle band. 
		The catalyst is constructed as carrying half of its mass uniformly distributed over $d_{\rm C}-1$ of its entries and the other half in a single entry. 
		This means that the middle band is divided into two subregions, middle left (ML) and middle right (MR), where the total probability mass coloured in blue equals the mass coloured in red. 
		Furthermore, each of these subregions has its mass uniformly distributed over its entries. 
		\emph{Bottom:} By construction, both the subregions BL and MR as well as ML and TR have the same number of entries. 
		Hence, we can swap BL and MR by means of a permutation, and similarly for ML and TR. 
		This permutation results in a reduced state on $ S $ of the form Eq.~\eqref{eq:extracted_work_state} and hence produces the claimed work extraction probability. 
	Moreover, it leaves the marginal state of the catalyst unchanged, so that the permutation induces a valid catalytic channel.}
	\label{fig:micro} 
\end{figure*}

\begin{proof}
	{A sketch of the proof is given in Fig.~\ref{fig:micro}. The proof is constructive in the sense that we provide an explicit catalyst and unitary. 
		We first introduce some useful notation. {Define $g := g(I)$, $g_+:= g(I_+)=g^2$ and let $P_{\rm S}(I)$ and $P_{\rm S}(I_+)$ be the projectors onto the corresponding energy subspaces}.
		Let $\ket{E_-}$ be \emph{any} eigenstate of the Hamiltonian such that {$0 \leq E_-/N = e_- \leq e$}. 
		Following this notation, the initial state of the system is
		\begin{align}
			\Omega_{\rm S}(I) = \frac{1}{g} P_{\rm S}(I). 
		\end{align}
		The aim is to bring the system to a state that is an equal mixture of {$\proj{E_-}$ and $\Omega(I_+)$}.
		To do this, we employ a catalyst of dimension $d_{\rm C} = g+1$. {Let $\{\ket{i}_{\rm C}\}_{i=1}^{d_{\rm C}}$ be an arbitrary orthonormal basis on the Hilbert-space of the catalyst and let $P_{\rm C}=\sum_{i=1}^g \proj{i}$. The initial state on the catalyst is given by
			\begin{align}
				\sigma = \frac{1}{2g} P_{\rm C} + \frac{1}{2}\proj{d_{\rm C}}_{\rm C}. 
			\end{align}
	}}
	{We define the unitary $U$ by the conditions
		\begin{align}
			U [P_{\rm S}(I)\otimes \proj{d_{\rm C}}_{\rm C}] U^\dagger & = \proj{E_-} \otimes P_{\rm C}                     \\
			U [P_{\rm S}(I)\otimes P_{\rm C}] U^\dagger                & = P_{\rm S}(I_+) \otimes \proj{d_{\rm C}}_{\rm C}. 
		\end{align}
		This is possible since i) the corresponding subspaces have the same dimension, ii) the subspaces in the two equations are orthogonal and iii) subspaces of the same dimension can always be mapped into each other by a unitary. In fact there will be many different unitaries achieving this, and any of them is fine for our purposes.
				
		Applying $ U $ to the state $\Omega_{\rm S}(I)\otimes \sigma_{\rm C}$ one obtains
		\begin{widetext}
			\begin{align}
				\label{eq:effect_unitary1}U \left(\Omega_{\rm S}(I) \otimes \sigma_{\rm C} \right)U^{\dagger} & =\frac{1}{2g^2} U \left( P_{\rm S}(I) \otimes  P_{\rm C} \right) U^{\dagger}+ \frac{1}{2g} U \left( P_{\rm S}(I) \otimes\frac{1}{2}\proj{d_{\rm C}}_{\rm C} \right) U^{\dagger} \\
				\label{eq:effect_unitary2}                                                                    & =\frac{1}{2} \Omega_{\rm S}(I_+) \otimes \proj{d_{\rm C}}_{\rm C}+ \frac{1}{2g} \proj{E_-}\otimes P_{\rm C}.                                                                    
			\end{align}
		\end{widetext}
		{It is clear from~\eqref{eq:effect_unitary2} that 
			\begin{align}
				\tr_{\rm S}(U (\Omega_{\rm S}(I)\otimes \sigma_{\rm C}) U^{\dagger})=\sigma_{\rm C}, 
			\end{align} 
			as required for a catalytic channel. Moreover, the quantity of interest $ P(w\geq e-e_-) $ given by this channel $ \mathcal{C} $ (defined by $ U $ and $ \sigma_{\rm C} $) can be derived by noting that
			\begin{align} \label{eq:extracted_work_state}
				\mathcal{C}(\Omega_{\rm S}(I))=\frac{1}{2}\Omega(I_+) + \frac{1}{2} \proj{E_-}, 
			\end{align}
			so that 
			$     p(W \geq e- E'/n) = \frac{1}{2} $.}}
\end{proof}

\section{Proof of Proposition~1 in the main text}
\label{app:canonical}
In this section, we provide the proof and full statement of Proposition~1 in the main text. 
This proof is very similar to that of the micro-canonical case presented in the previous section, we will hence only describe the adjustments that have to be made. Also, unlike in Appendix~\ref{app:microcanonical}, we now again consider the standard setting and definition of catalytic channels as introduced in the main text.
In the following, we denote by $P_{\rm S}(I)$ the projector onto a specific energy-window $I$. Then $g(I)$ is equal to the rank of $P_{\rm S}(I)$.
We consider Hamiltonians $H_S^{(N)}$ on a regular lattice $\Lambda^{(N)}$ of $N$ sites and assume that the $H_S^{(N)}$ (for different values of $N$) constitute a sequence of \emph{local}, \emph{uniformly bounded} Hamiltonians:
\begin{align}
	H_S^{(N)} = \sum_{x \in \Lambda^{(N)}} h_x, 
\end{align}
where each term $h_x$ acts on sites at most a distance $l$ away from $x$ and the norm of each term is bounded as $\norm{h_x}\leq h$ independent of the system size for some constant $h$.

\begin{proposition}[Lower bound to the probability of work extraction]\label{prop:gibbs}
	Fix an inverse temperature $\beta>0$ and consider a sequence of local, uniformly bounded $N$-particle Hamiltonians $H_S^{(N)}$ on a regular, $D$-dimensional lattice. Assume that the states $\omega_{\beta}(H_S^{(N)})$ have a finite correlation length bounded by a constant and denote by $e^{(N)}$ the energy density corresponding to $\beta$. Let $\delta>0$ be fixed and consider $I^{(N)} := [e^{(N)} N-\delta\sqrt{N},e^{(N)} N]$. 
	Further assume {that} there exist micro-canonical energy windows $I^{(N)}_+$ with $g(I^{(N)}_+) = g(I^{(N)})^2$.
	Then, for sufficiently large $N$, there exists, 
	for any $0 < e_- < e^{(N)}$, a corresponding sequence of catalytic channels such that 
	\begin{align}
		p(w\geq e^{(N)} - e_-) \geq 1/2 - C \e^{-\frac{\left(\delta^2 N\right)^{\frac{1}{1+D}}}{K}}, 
	\end{align}
	where $C,K>0$ are constants. 
\end{proposition}
Before giving the proof, we again emphasize the weakness of the {assumptions in the statement}, which, in the limit of large $N$, can be satisfied to arbitrary precision if the density of states grows exponentially within $I^{(N)}$, as is typically the case. Furthermore, let us emphasize that the energy densities $e^{(N)}$ fluctuate arbitrarily little (for sufficiently large $N$) from a constant $e$ due to the locality of temperature \cite{Kliesch2014}.
\begin{proof}
	The proof follows the proof for the micro-canonical case in Appendix~\ref{app:microcanonical}. In particular, the unitary that we use is exactly the same as that constructed in the proof for the micro-canonical case.
	However, here we do not construct the state of the catalyst explicitly, but allude to Lemma~\ref{lemma:catfromU}, which ensures there is always some catalyst given the unitary that we consider.
	What remains to be done is to show that for every such catalyst the probability distribution of work is as claimed. 
	To do this, we denote by $r$ the initial probability of an energy-window $I$ in the initial thermal state given by
	\begin{align}
		{r(I)} = \tr(P_{\rm S}(I) \omega_\beta(H_{\rm S})) 
	\end{align}
	{and by $r_- = \bra{E_-}\omega_\beta(H_{\rm S}) \ket{E_-}$ the initial weight on the low-energy eigenstate $\ket{E_-}$.} Here and in the following, we drop the explicit dependence on the system-size for simplicity of notation. The following arguments work as long as $N$ is large enough such that the energy-windows $I$ and $I_+$are disjoint. 
	Denote by $\{q_i\}_{i=1}^{d_C}$ the spectrum of the catalyst.
	By considering the action of the used unitary, it is easy to see that a necessary condition for the transition being catalytic under the given unitary is that
	{	\begin{align} \label{eq:constraint_column}
		q_{d_C}\left(r(I) + r(I_+)\right) = (1-q_{d_C})(r(I)+r_-).
		\end{align}}
	This can be seen, for example, from Fig.~\ref{fig:micro}, where the above represents the condition of catalyticity for the right-most column. 
	Solving in \eqref{eq:constraint_column} for $q_{d_C}$, we find that
	{	\begin{align}
		q_{d_C} = \frac{r(I)+r_-}{2r(I) + r_- + r(I_+)}.
		\end{align}}
	We now invoke the result from Ref.\ \cite{Anshu2016} (as previously in 
	the proof of Lemma~\ref{lemma:NMW-GP}) which implies that 
	\begin{align}
		{ r(I)} \geq  1 - \epsilon_N, 
	\end{align}
	where there exist constants $ C,K $ such that
	\begin{align}
		\epsilon_N \leq C e{-\frac{\left(\delta^2 N\right)^{\frac{1}{1+D}}}{K}} . 
	\end{align}
	For large enough $N$, the energy windows $I$ and $I_+$ are disjoint. Hence $0\leq {r_- + r(I_+) \leq 1 - r(I)}$ and  we find {
		\begin{align}
			q_{d_C} & \geq \frac{r(I)}{2r(I) + 1- r(I)} = \frac{r(I)}{1+r(I)}        \\
			        & \geq \frac{r(I)}{2} \geq \frac{1}{2}\left(1-\epsilon_N\right). 
		\end{align}}
	Finally, we find
	\begin{align}
		p(w\geq e-e_-) & \geq P({E_f = E_-} | E_i \in I)w(I) = q_{d_c} \cdot {r(I)} \nonumber \\
		               & \geq\frac{1}{2}(1-\epsilon_N)^2 \geq \frac{1}{2} - \epsilon_N.       
	\end{align}
\end{proof}

\begin{lemma}[Existence of catalysts]\label{lemma:catfromU} Let $\rho_{\rm S}$ be a quantum state on a finite-dimensional Hilbert-space $\mc H_{\rm S}$ and $U$ be a unitary on the Hilbert-space $\mc H_{\rm S}\otimes \mc H_{\rm C}$, where $\mc H_{\rm C}$ is an arbitrary finite-dimensional Hilbert-space. Then there exists a density matrix $\sigma_{\rm C}$ such that
	\begin{align}
		\tr_{\rm S}\left(U ( \rho_{\rm S}\otimes \sigma_{\rm C} ) U^\dagger\right) = \sigma_{\rm C}. 
	\end{align}
\end{lemma}
\begin{proof}
	The map $\sigma_{\rm C} \mapsto \tr_{\rm S}\left(U (\rho_{\rm S}\otimes \sigma_{\rm C}) U^\dagger\right)$ specifies a quantum-channel. Since every quantum channel is a continuous map on the compact and convex set of states, it has a fixed point by Brouwer's fixed point theorem (\cite{watrous2018theory}, Section 4.2.2). 
\end{proof}

\section{Proof of Proposition~3 in the main text}
\label{app:many_player_strategies}
\begin{figure}[t]
	\centering
	\includegraphics[width=0.5\textwidth]{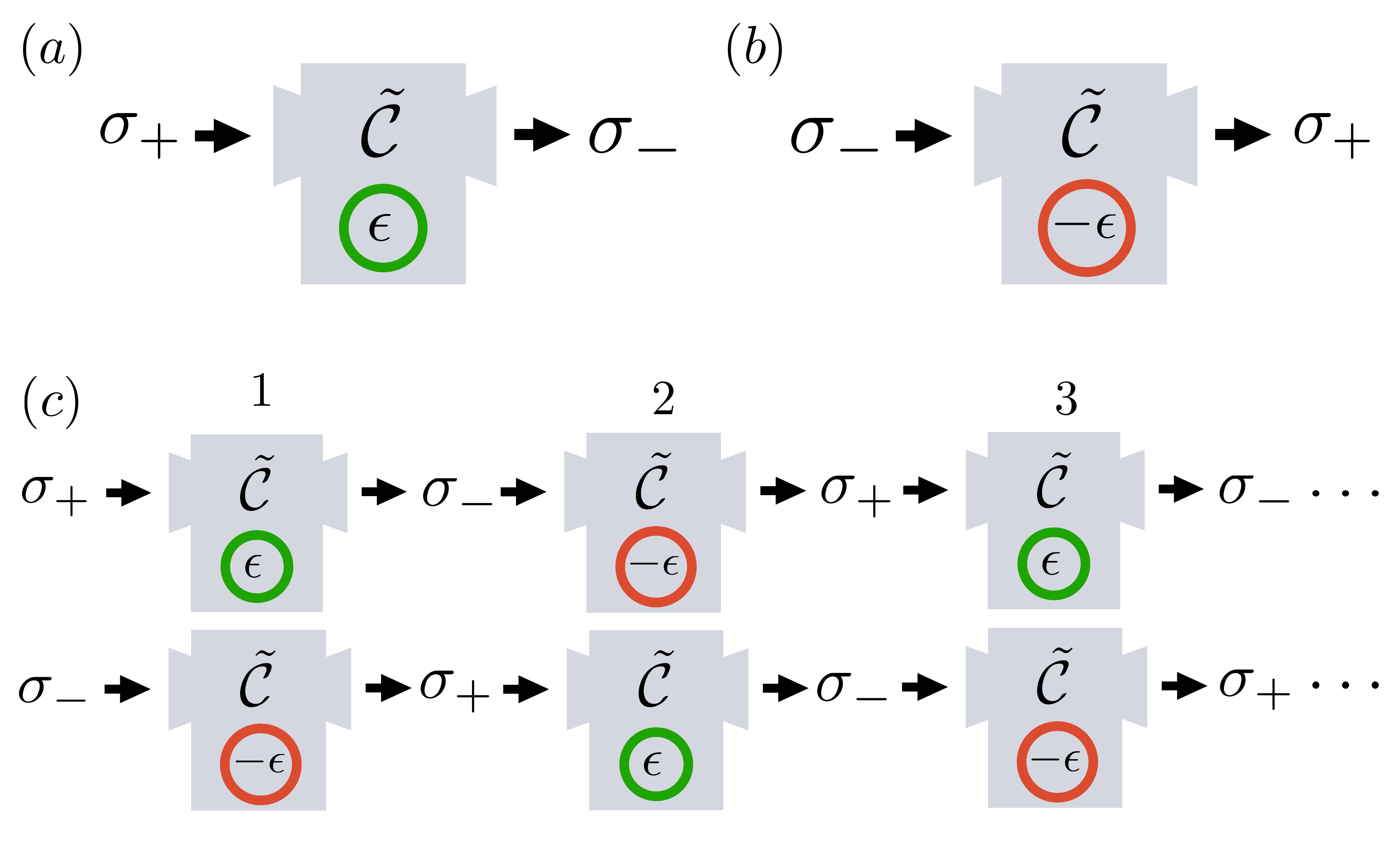}
	\caption{The idea behind the proof of Proposition~3 in the main text: For any choice of unitary, we can understand the second condition in the Def.~2 of the main text as the definition of a quantum channel $\tilde{\mc C}$ acting on $C$. We find a $\tilde{\mc C}$ and two states $\sigma_-, \sigma_+$ with the following properties: (a) If the initial state of the catalyst is $\sigma_+$, the result of running the standard protocol is to extract positive work $\epsilon$ from the system, while the state of the catalyst is changed to $\tilde{\mc C}(\sigma_+) = \sigma_- $. 
		(b) The \emph{same} unitary, however, for initial state $\sigma_-$, extracts negative work $- \epsilon$ and changes the catalyst state to $\tilde{\mc C}(\sigma_-) = \sigma_+$. 
		(c) Hence, if we initialize the catalyst in the state $\sigma = \frac{1}{2}(\sigma_+ + \sigma_-)$, then there are two ``branches'' of  work extraction distributions, each occurring with probability $1/2$, while the resulting channel on $S_i$ is catalytic for every $i$.
		Note that, if the agent \emph{knew} whether her input state was $\sigma_+$ or $\sigma_-$, then she could condition her unitary $U$ on this knowledge and achieve the claimed work distribution easily. Hence, the key achievement of the proof is to show that agents can achieve correlated work distributions \emph{without knowing the initial state of the catalyst.}
	}
	\label{fig:multi}
\end{figure}

Proposition~3 in the main text follows straightforwardly once we realize that we can tune the process used in the construction of the proof for Proposition~1 in the main text in such a way that its repeated application implies the claimed work distribution. This follows because we have great freedom in choosing the state $E_-$. In particular, in terms of notation of the previous section, let $e_+^{(N)}$ denote the energy density around which the window $I_+^{(N)}$ is centered. Then we choose $E_-$ in such a way that $e - E_-/N = e_+ - e$ to ensure that the extracted and invested amount of work in every iteration are exactly the same. The above choice of $E_-$ is always possible for the Hamiltonians with exponentially growing density of states that we consider (for which $e_+$ will not be much greater than $e$.)

Fig.~\ref{fig:multi} provides a sketch of the proof. For the many-player process described in the main text, let 
\begin{align}
	p(w_2, w_3, w_4, \dots | w_1) 
\end{align}
denote the work probability distribution for agents $2$ to $n$ conditional on the player $1$ extracting work $w_1$. The key recognition then is that, for any $n$, by construction of the catalytic channel, 
\begin{equation} \label{eq:maximal_correlation}
	p(w_2, w_3, \dots | w_1) = 1
\end{equation}
whenever $w_{i} = - w_{i-1}$ for all $i \in \{2, \dots, n\}$, while 
\begin{equation}
	p(w_2, w_3, \dots | w_1) = 0
\end{equation}
in all other cases.
This is because, if the extracted work in the first round was negative, corresponding to an increase in the system's energy, then by construction of the unitary, the final state of the catalyst is $\sigma' = \proj{d}$ with probability one, since all transitions that lead to an increase in energy on the system result in this final state. This, in turn, is sufficient to determine that, for the second player, the application of the unitary to this catalyst state $\sigma'$ and her copy of $\omega_\beta(H)$ will result in a decrease of the system's energy (and hence positive work extraction) and a final catalyst state $\sigma''$ with support on the subspace $\sum_{i}^g \proj{i}$, etc. This reasoning can be applied to an arbitrary number of agents and also to the case in which the extracted work in the first round was positive, and hence implies \eqref{eq:maximal_correlation}. The claimed work distributions then follow from 
\begin{align}
	p(w_1, w_2, \dots, w_n) = p(w_2, w_3, w_4, \dots | w_1) p(w_1), 
\end{align}
together with Proposition~1 in the main text.
We also note that a similar conclusion holds in the case of the microscopic toy-example presented in Section~\ref{app:microscopic_toy_example}, where this behaviour can be checked easily by explicit calculation.

\section{Proof of Proposition~2 in the main text} 
\label{app:lower_bound_on_catalyst_s_dimension}

Given a catalytic channel $\mc C$, let $\mc U$ denote the unitary channel applied to the joint system $SC$ when dilating the channel. The key observation is that, if $\mathcal U$ is unitary, then $\mathcal U^*$ is trace-preserving and hence maps quantum states to quantum states (in fact, this property holds for all unital channels). Here, $^*$ denotes the Hilbert-Schmidt adjoint. We then write
\begin{align}
	\langle \e^{\beta W}\rangle & = \Tr\left(\omega \mathcal{C}(\mathbf 1)\right)                                                               \\
	                            & = \Tr\left(\omega\otimes \mathbf{I} \mathcal{U}(\mathbf 1 \otimes \sigma)\right)                              \\
	                            & = d_C \Tr\left(\mathcal{U}^*\left(\omega\otimes \frac{\mathbf{I}}{d_C}\right) \mathbf 1 \otimes \sigma\right) \\
	                            & \leq d_C \norm{\mathbf 1\otimes \sigma}_\infty \norm{\frac{\mathbf 1}{d_C}\otimes \sigma}_1                   \\
	                            & = d_C \norm{\sigma}_\infty.                                                                                   
\end{align}
Here, the first equality is simply Eqn.~2 in the main text and we write $\omega$ instead of $\omega_\beta(H)$.
Similarly, we get
\begin{align}
	\langle \e^{\beta W}\rangle & = d\Tr\left(( \omega\otimes \mathbf{I} )\mathcal{U}\left(\frac{\mathbf 1}{d} \otimes \sigma\right)\right) \\
	                            & \leq d \norm{\omega \otimes \mathbf{I}}_\infty = d \norm{\omega}_\infty.                                  
\end{align}

\section{Non-trivial Hamiltonian on the catalyst}\label{sec:nontrivialH}
In this section we show that the probability distribution of work done on the system is independent of the Hamiltonian on the catalyst.
To do this, let us first assume we had a catalytic process that uses a catalyst with a non-trivial 
Hamiltonian $H_C$ and a quasi-classical state $\sigma_C$, i.e., $[H_C,\sigma_C]=0$. We assume that $\sigma_C$ is quasi-classical, since it is well known that it is impossible to associate a meaningful random variable of work in the case coherent initial states \cite{perarnau2017no}.
Using the two-time measurement process on the system and catalyst together, we can then associate a bi-partite work-distribution $P(W^{(S)},W^{(C)})$, where 
\begin{equation}
	W^{(S)}=E_f^{(S)} - E_i^{(S)}
\end{equation}
denotes the work done on the system and 
\begin{equation}
	W^{(C)}=E_f^{(C)}-E_i^{(C)}
\end{equation} the work done on the catalyst. 
The work distribution on the system is simply given by the marginal
\begin{align}
	P\left(W^{(S)}\right) = \int P\left(W^{(S)},W^{(C)}\right) \, \mathrm d W^{(C)}. 
\end{align}
Let us write $\sigma_C = \sum_j \sigma_j |E^{(C)}_j\rangle\langle E^{(C)}_j|$ and $\omega_\beta(H) = \sum_k w_k |E^{(S)}_k\rangle\langle E^{(S)}_k|$. 
We then get
\begin{widetext}
	\begin{align}
		P\left(W^{(S)}\right) & = \sum_{E^{(S)}_f-E^{(S)}_i=W^{(S)}} \sum_{E^{(C)}_{f'}}  \sum_{E^{(C)}_{i'}}  P\left(E^{(S)}_f, E^{(C)}_{f'} | E^{(S)}_i, E^{(C)}_{i'} \right)P(E^{(S)}_i)P(E^{(C)}_{i'})  \nonumber                                                                                                                                            \\
		                      & = \sum_{E^{(S)}_f-E^{(S)}_i=W^{(S)}} \sum_{E^{(C)}_{f'}} \sum_{E^{(C)}_{i'}} \langle E^{(S)}_f|\otimes\langle E^{(C)}_{f'}| \left(U \left(w_i \sigma_{i'} |E^{(S)}_i\rangle\langle E^{(S)}_i|\otimes|E^{(C)}_{i'}\rangle\langle E^{(C)}_{i'}|  \right) U^\dagger\right)|E^{(S)}_f \rangle \otimes |E^{(C)}_{f'}\rangle \nonumber \\
		                      & = \sum_{E^{(S)}_f-E^{(S)}_i=W^{(S)}} \langle E^{(S)}_f|\tr_C\left(U \left(w_i |E^{(S)}_i\rangle\langle E^{(S)}_i|\otimes \sigma \right) U^\dagger\right)|E^{(S)}_f\rangle                                                                                                                                                        \\
		                      & = \sum_{E^{(S)}_f-E^{(S)}_i=W^{(S)}} \langle E^{(S)}_f|\mathcal C \left(w_i |E^{(S)}_i\rangle\langle E^{(S)}_i|\right)|E^{(S)}_f\rangle.                                                                                                                                                                                         
	\end{align}
\end{widetext}
It is hence identical with the one obtained on the system alone when we think of the catalyst as a system with a trivial Hamiltonian, that is, with the distribution as defined above Eqn.~2 in the main text. 
This shows that we can always assume that the catalyst has a trivial Hamiltonian, in which case it is clear that no energy flows from the catalyst to the system, even probabilistically.
Therefore, such an energy flow is not necessary to implement catalytic transitions.

\section{Comparison with literature on generalized Jarzynski equalities in the presence of correlations}
\label{app:sagawa}

In recent years, the role of correlations, specifically quantified by the mutual information, has been well studied, in particular with respect to its influence on the Jarzynski equality~\cite{sagawa2012fluctuation,sagawa2013role}, even leading up to experimental demonstrations to test these theoretical results~\cite{toyabe2010experimental,koski2014experimental}. One may ask, how the results of this manuscript fit in the context of that line of research. This section provides a brief overview of the main differences.

In Ref.~\cite{sagawa2012fluctuation} the core observation is that the presence of initial correlations between a system S and an ancillary (catalyst) C can be used to create a thermodynamic advantage, in the sense that such processes obey a generalized JE and Second law and hence can be used to by-pass the constraints imposed by the original JE and Second law.
Specifically, \cite{sagawa2012fluctuation} derives (according to their generalized version of Jarzynski equality) a bound on the work \textit{performed} on the system that is given by
\begin{equation}
	\langle W \rangle \geq \langle\Delta F\rangle + \langle\Delta E_{\rm int}\rangle + \beta^{-1}\langle\Delta I\rangle,
\end{equation}
where $ \langle\Delta F\rangle $ is the difference between final and initial equilibrium free energy on the system, $ \langle\Delta E_{\rm int}\rangle $ for the energy difference coming from the interaction Hamiltonian between system and catalyst, and finally $ \langle\Delta I\rangle $ is the change in mutual information between system and catalyst. 
For our setup, both $ \langle\Delta F\rangle $ and $ \langle\Delta E_{\rm int} \rangle $ are zero. Given that the extracted work $ W_{\rm ext} = -W $, the above bound reduces to
\begin{equation}\label{eq:work_inequality_sagawa}
	\langle W_{\rm ext}\rangle \leq -\beta^{-1} \langle\Delta I\rangle,
\end{equation}
which says that if one allows the consumption of mutual information (leading to $ \Delta I <0 $), then it is possible to violate the average second law, namely extract some positive amount of $ W_{\rm ext} $ from a Gibbs state, for instance by reducing the entropy of the system in the process.
This particular viewpoint of correlations (information) being a thermodynamic resource is a mature and well-studied one. 

In our setting, however, the initial state of system and catalyst are always uncorrelated, which means that we always have $\langle\Delta I\rangle \geq 0$. Hence it is clear that the type of catalytic operation studied 
in Ref.~\cite{sagawa2012fluctuation} cannot correspond to our setting, since the generalized JE and Second law allow for violations of the original JE and Second law \emph{only} if $\langle\Delta I\rangle < 0$. The difference to our setting, however, is easily understood. It lies in the fact that here we allow for more general joint evolutions of the system and the catalyst. Indeed, it is easy to see that under the requirement that the initial state between catalyst and system be uncorrelated, the channels that can be implemented on the system via the operations allowed in Ref.~\cite{sagawa2012fluctuation} are unital channels, for which we show above that they cannot be used to by-pass the JE (see Fig.~\ref{fig:diagram_violation}). This is because in the above works, the catalyst is required to not evolve over time. In contrast, the notion of a $\beta$-catalytic channel allows for the evolution of the catalyst to be non-trivial, as long as the final density matrix describing the catalyst is unchanged. Since this constraint only requires the \emph{statistical} invariance of the catalyst, this allows for a much broader class of evolutions to be implemented on the system and hence explains how we can by-pass the JE and NMW in a setting where the marginal entropy of the system has to increase. In summary, the key differences to the line of work rooted in Refs.~\cite{sagawa2012fluctuation,sagawa2013role} are that we study processes that by-pass the JE by means of the \emph{creation} of correlations paired with catalysts that evolve non-trivially over time, while in the above work processes are studied that by-pass the JE by means of the \emph{absorption} of initial correlations paired with catalysts that do not evolve over time.

\section{Is it necessary to establish correlations with the catalyst?}\label{app:trumping}
In our definition of $\beta$-catalytic channels, we allow the catalyst to become correlated with the system. These correlations are certainly necessary for the correlated multi-player strategies discussed in the main text, but one might wonder whether they are also necessary to violate NMW on a single system.
To make this question concrete, consider the set of \emph{$\beta$-trumping channels}, where a quantum channel $ \mathcal{T} $ is in this set iff it has the form
\begin{equation}
	\mathcal{T} (\rho) = \tr_2 (\mathcal{N} (\rho\otimes\sigma)) , 
\end{equation}
where $ \mathcal{N} (\mathbb{I})=\mathbb{I} $ is unital and $ \mathcal{N} (\omega_\beta(H)\otimes\sigma) = \rho'\otimes\sigma$.
Note that in the case of $\beta$-catalytic channels, we restricted the corresponding channel $\mathcal N$ to be unitary. Here, we allow instead for the more general class of unital channels. 
We will prove that in the unitary case, NMW cannot be violated by $\beta$-trumping channels even though Jarzynski's equality may be violated. We will also present arguments that suggest that the same is true in the unital case. 

It is worth noting, for starters, that the fully thermalizing channel is exactly a $\beta$-trumping channel where $ \sigma = \omega_\beta(H)$, and $ \mathcal{N} $ is a unitary swap between the system and catalyst. 
Thus, even in the case of a unitary channel $\mathcal N$, such $\beta$-trumping channels can violate Jarzynski's equality. 
On the other hand, in the main text we have demonstrated that the thermalizing channel cannot violate NMW since it is Gibbs preserving. Hence, the above leaves open the question whether the NMW condition can be violated by means of $\beta$-trumping channels. However, we do not believe that this is the case, for the following reasons:

i) Our constructions of violating NMW can \emph{not} work in the trumping case. This is because in the trumping setting the so-called min-entropy $S_\infty$ (minus $\log$ of the largest eigenvalue) of the final state has to be at least as large as that of the initial state (see for example Ref.~\cite{Brandao2015}). However, in our constructions, the final min-entropy is essentially given by $-\log(p(w\geq \epsilon) \approx \log(2)$, whereas the initial min-entropy is extensive in $N$. It thus \emph{decreases} by a macroscopic amount.

ii) The previous point also suggests a route for arguing that $\beta$-trumping channels cannot be used to violate NMW in general: 
We now present an argument that rules out violations of NMW in the case of a microcanonical initial state $\Omega$ with energy density $e$, but we expect that similar statements hold true for the canonical case due to equivalence of ensembles-type of arguments.
Because of the highly peaked probability distribution of the energy density for a macroscopic, non-critical many-body system, it is easy to see that the probability $p(w\geq \epsilon)$ to extract work per particle at least $\epsilon$ is (up to arbitrarily small corrections for large $N$) given by the total probability of measuring an energy below $(e-\epsilon) N$ in the final state $\mathcal T(\Omega)$. Let us denote the projector onto these energies by $P$. We then have
\begin{align}
	p(w\geq \epsilon) \approx \Tr[P\mathcal T(\Omega)], 
\end{align}
where the approximation is arbitrarily good as $N\rightarrow \infty$. 
This insight also was an essential ingredient to the proof that Gibbs-preserving maps cannot violate NMW. 
Now, to leading order, the total number of states with energy below $(e-\epsilon)N$ is given by $\exp(s(e-\epsilon)N)$, where $s(e-\epsilon)$ is the microcanonical entropy density at energy density $e-\epsilon$. Since the total weight in this subspace is $p(w\geq \epsilon)$, the final min-entropy is upper bounded by
\begin{align}
	S_{\min}^{\mathrm{(final)}} \leq -\log(p(w\geq \epsilon)) + s(e-\epsilon) N. 
\end{align}
However, since trumping requires $S_{\min}^{\mathrm{(final)}}\geq S_{\min}^{\mathrm{(initial)}} = s(e)N$, we then find
\begin{align}
	p(w\geq \epsilon) \leq \exp(-(s(e)-s(e-\epsilon))N) \rightarrow 0, 
\end{align}
as $N\rightarrow \infty$ for any $\epsilon>0$. This shows that NMW holds for $\beta$-trumping channels in the micro-canonical case. Note that when we allow the catalyst to become correlated, NMW can be violated for microcanonical initial states, as shown above. This already makes clear that correlated catalysts provide a strict advantage in this set-up.    

iii) Finally, let us also show that if we assume that $\mathcal N$ is \emph{unitary}, as we do in the case of $\beta$-catalytic channels, then NMW cannot be violated if the catalyst remains uncorrelated. 
The reason is the following: Since the global transformation on system and catalyst is unitary, it leaves the spectrum invariant. Since the catalyst remains invariant and uncorrelated, this implies that already the spectrum of the initial density matrix on the system has to remain invariant. Therefore there exists a unitary $V$, such that $\mathcal T[\omega_\beta(H)] = V \omega_\beta(H) V^\dagger$. 
As argued in case ii), we then have
\begin{align}
	p(w\geq \epsilon) & \approx \Tr[P \mathcal T[\omega_\beta(H)]]= \Tr[P V\omega_\beta(H)V^\dagger] \\ &\leq \Tr[P\omega_\beta(H)],
\end{align}
where the last inequality follows because Gibbs states are passive states and the first approximation holds to arbitrary accuracy as $N\rightarrow \infty$.
However, by the same concentration inequalities we used to prove of our main results, we have
\begin{align}
	\Tr[P\omega_\beta(H)] \leq K \exp(- k (\epsilon^2 N)^{1/(1+D)}), 
\end{align}
for a non-critical many-body system in $D$ spatial dimensions (with constants $k,K>0$). Thus, NMW holds true in this case as well.  

\end{document}